\documentclass[a4paper,UKenglish,cleveref, autoref, thm-restate]{lipics-v2021}

\usepackage{bm}
\usepackage{cleveref}
\usepackage{mathtools}
\usepackage{booktabs}
\usepackage{xspace}
\usepackage{multicol}
\usepackage{tikz}
\usetikzlibrary{automata, petri, positioning, shapes, arrows.meta}
\usepackage{macros}
\hideLIPIcs
\usepackage{thm-restate}

\title{Population Protocols with Unordered Data}

\author{Michael Blondin}{Department of Computer Science, Universit\'e de Sherbrooke, Canada}{michael.blondin@usherbrooke.ca}{https://orcid.org/0000-0003-2914-2734}{Supported by a Discovery Grant from the Natural Sciences and Engineering Research Council of Canada (NSERC), and by the Fonds de recherche du Québec – Nature et technologies (FRQNT).}

\author{Fran\c{c}ois Ladouceur}{Department of Computer Science, Universit\'e de Sherbrooke, Canada}{francois.ladouceur@usherbrooke.ca}{https://orcid.org/0009-0000-7651-6685}{Supported by a scholarship from the Natural Sciences and Engineering Research Council of Canada (NSERC), and by the Fondation J.A. De S\`eve.}

\authorrunning{M. Blondin and F. Ladouceur} 

\Copyright{Michael Blondin and Fran\c{c}ois Ladouceur} 

\ccsdesc[500]{Theory of computation~Distributed computing models}
\ccsdesc[500]{Theory of computation~Automata over infinite objects}

\keywords{population protocols, unordered data, colored Petri nets}

\category{} 

\relatedversion{} 

\acknowledgements{We thank Manuel Lafond for his ideas and feedback in the early phase of our research. We further thank the anonymous reviewers for their comments and insightful suggestions.}

\nolinenumbers 

\EventEditors{John Q. Open and Joan R. Access}
\EventNoEds{0}
\EventLongTitle{50th EATCS International Colloquium on Automata, Languages and Programming (ICALP 2023)}
\EventShortTitle{ICALP 2023}
\EventAcronym{ICALP}
\EventYear{2023}
\EventDate{July 10--14, 2023}
\EventLocation{Paderborn, Germany}
\EventLogo{}
\SeriesVolume{0}
\ArticleNo{0}

\begin{document}

\maketitle

\begin{abstract}
  Population protocols form a well-established model of computation of
  passively mobile anonymous agents with constant-size memory. It is
  well known that population protocols compute Presburger-definable
  predicates, such as absolute majority and counting predicates. In
  this work, we initiate the study of population protocols operating
  over arbitrarily large data domains. More precisely, we introduce
  \emph{population protocols with unordered data} as a formalism to
  reason about anonymous crowd computing over unordered sequences of
  data. We first show that it is possible to determine whether an
  unordered sequence from an infinite data domain has a datum with
  absolute majority. We then establish the expressive power of the
  ``immediate observation'' restriction of our model, namely where, in
  each interaction, an agent observes another agent who is unaware of
  the interaction.
\end{abstract}

\section{Introduction}\label{sec:intro}
\subparagraph*{Context.}

Population protocols form a well-established model of computation of
passively mobile anonymous agents with constant-size
memory~\cite{AADFP06}. Population protocols allow, \eg, for the formal
analysis of chemical reaction networks and networks of mobile sensors
(see~\cite{MS18} for a review article on population protocols and more
generally on dynamic networks).

In a population protocol, anonymous agents hold a mutable state from a
finite set. They collectively seek to evaluate a predicate on the
initial global state of the population. At each discrete moment, a
scheduler picks two agents who jointly update their respective
states according to their current states. Such a scheduler is assumed
to be ``fair'' (or, equivalently, to pick pairs of agents uniformly at
random). Let us illustrate the model with a classical protocol for the
aboslute majority predicate. Consider a population of $\ell$ (anonymous)
agents, each initialized with either $Y$ or $N$, that seek to compute
whether the number of $Y$ exceeds the number of $N$, \ie, to
collectively evaluate the predicate $\varphi(\#Y, \#N) \defeq (\#Y >
\#N)$. For example, a population of $\ell = 5$ agents may be
initialized to $\multiset{Y, N, Y, Y, N}$. An update of two agents
occurs according to these four rules:
\begin{center}
  \begin{tabular}{lcl}
    \toprule
    \emph{strong to weak} &
    \emph{propagation of winning side} &
    \emph{tiebreaker} \\
    \midrule
    \multirow{2}{*}{$\multiset{Y, N} \to \multiset{n, n}$} &
    $\mathmakebox[30pt][r]{\multiset{Y, n}} \to \mathmakebox[30pt][l]{\multiset{Y, y}}$ &
    \multirow{2}{*}{$\multiset{y, n} \to \multiset{n, n}$} \\
    & $\mathmakebox[30pt][r]{\multiset{N, y}} \to \mathmakebox[30pt][l]{\multiset{N, n}}$ & \\
    \bottomrule
  \end{tabular}
\end{center}
A possible execution from the aforementioned population is
$
\multiset{Y, N, Y, Y, N} \allowbreak
\reach{} \allowbreak
\multiset{Y, N, Y, n, n} \allowbreak
\reach{} \allowbreak
\multiset{Y, n, n, n, n}
\reach{}
\multiset{Y, y, n, n, n}
\reach{}
\cdots
\reach{}
\multiset{Y, y, y, y, y}$.

Agents in states $\{Y, y\}$ believe that the output of $\varphi$
should be $\true$, while agents in $\{N, n\}$ believe that it should
be $\false$. Thus, in the above execution, a lasting $\true$-consensus
has been reached by the population (although no agent is locally
certain of it).

It is well known that population protocols compute precisely the
predicates definable in Presburger arithmetic, namely first-order
logic over the naturals with addition and order. This was first shown
through convex geometry~\cite{AAER07}, and reproven using the theory
of vector addition systems~\cite{EGLM17}. For example, this means
that, given voting options $\{1, \ldots, k\}$, there is a population
protocol that determines whether some option $i$ has an absolute majority, \ie,
whether more than $\ell / 2$ of the $\ell$ agents initially hold a
common $i \in \{1, \ldots, k\}$.

Since $k$ must be stored in the state-space, such a majority protocol
can only handle a fixed number of voting options. As rules also depend on
$k$, this means that a whole population would need to be reconfigured
in order to handle a larger $k$, \eg\ if new voting options are made
available. This is conceptually impractical in the context of flocks
of anonymous mobile agents. Instead, we propose that the input of an
agent can be modeled elegantly as drawn from an infinite set $\D$,
with rules independent from $\D$.

\subparagraph*{Contribution.}

In this work, we initiate the study of population protocols operating
over arbitrarily large domains. More precisely, we propose a more
general model where each agent carries a read-only datum from an
infinite domain $\D$ together with a mutable state from a finite
set. In this setting, a population can, \eg, seek to determine whether
there is an absolute majority datum. For example, if $\D \defeq \{1, 2, 3,
\ldots\}$, then the population initialized with $\multiset{(1, x), (1,
  x), (2, x), (3, x), (1, x)}$ should reach a lasting
$\true$-consensus, while it should reach a lasting $\false$-consensus from
$\multiset{(1, x), (1, x), (2, x), (3, x), (2, x)}$.

As in the standard model, a fair scheduler picks a pair of agents. An
interaction occurs according to a rule of the form $\multiset{p,
q} \xrightarrow{d \sim e} \multiset{p', q'}$, where $d, e \in \D$ are
the data values of the two agents, and where ${\sim} \in
\{=, \neq\}$ compares them. As for states, we assume that arbitrarily
many agents may be initialized with the same datum; and that agents
can only compare data through (in)equality. So, $\D$ is not a set of
(unique) identifiers and hence agents remain anonymous as in the
standard model.

To illustrate our proposed model of computation, we first show that it
can compute the absolute majority predicate. This means that a \emph{single}
protocol can handle \emph{any} number of options in an absolute majority
vote. From the perspective of distributed computing, this provides a
framework to reason about anonymous crowd computing over unordered
sequences of data. From the standpoint of computer-aided verification,
this opens the possibility of formally analyzing single protocols
(\eg\ modeled as colored Petri nets) rather than resorting to
parameterized verification, which is particularly difficult in the
context of counter systems.

As a stepping stone towards pinpointing the expressive power of
\emph{population protocols with unordered data}, we then characterize
\emph{immediate observation} protocols. In this well-studied
restriction, rules have the form $\multiset{p, q} \xrightarrow{d \sim
  e} \multiset{p, q'}$, \ie\ an agent updates its state by observing
another agent (who is unaware of it). In standard population
protocols, this class is known to compute exactly predicates from
$\textbf{COUNT}_{*}$~\cite{AAER07}. The latter is the Boolean closure
of predicates of the form $\# q \geq c$, where $\# q$ counts the
number of agents in state $q$, and $c \in \N$ is a constant. In our
case, we show that immediate observation protocols compute exactly
\emph{interval predicates}, which are Boolean combinations of such
\emph{simple interval predicates}:
\begin{align}
  \exists \text{ pairwise distinct } d_1, d_2, \ldots, d_n \in \D :
  \bigwedge_{i = 1}^n \bigwedge_{j = 1}^m \# (d_i, q_j) \in T(i,
    j),\label{eq:count}
\end{align}
where $\# (d_i, q_j)$ counts agents in state $q_j$ with datum $d_i$,
and each $T(i, j) \subseteq \N$ is an interval.

In order to show that immediate observation protocols do not compute
more than interval predicates, we exploit the fact that (finitely
supported) data vectors are well-quasi-ordered. While our approach is
inspired by~\cite{AAER07}, it is trickier to simultaneously deal with
the several sources of unboundedness: number of data values, number of
agents with a given datum, and number of agents with a given state. As
a byproduct, we show that the absolute majority predicate cannot be computed by
immediate observation protocols.

To show the other direction, \ie\ that interval predicates are
computable by immediate observation protocols, we describe a protocol
for simple interval predicates. In contrast with the standard setting,
we need to implement existential quantification. This is achieved by a
data leader election and a global leader election. We call the latter
elected agent the ``controller''. Its purpose is to handle the
bookkeeping of data leaders choosing their role in~\eqref{eq:count}. A correction mechanism is carefully implemented
so that the population only reaches a $\true$-consensus upon locking a
correct assignment to the existential quantification.

\subparagraph*{Related work.}

It has been observed by the verification and concurrency communities
that population protocols can be recast as Petri nets. In particular,
this has enabled the automatic formal analysis of population
protocols~\cite{EGLM17,BEJM21} and the discovery of bounds on their
state complexity~\cite{CGHE22,BEGHJ20}. Our inspiration comes from the
other direction: we introduce protocols with data by drawing from the
recent attention to colored Petri
nets~\cite{GSAH19,HLT17,HLLLST16}. Our model corresponds to unordered
data Petri nets where the color and number of tokens is invariant.

Population protocols for computing majority and plurality have been
extensively studied (\eg,
see~\cite{DEGSUS21,BBHK22,BEFKKR18,BBBEHKK22} for recent results). To
the best of our knowledge, the closest work is~\cite{GHMSS16}, where
the authors propose space-efficient \emph{families} of deterministic
protocols for variants of the majority problem including plurality
consensus. They consider the $k$ voting options as ``colors''
specified by $\lceil \log k \rceil$ bits stored within the agents.

Other incomparable models of distributed systems with some sort of
data include broadcast networks of register automata~\cite{DST16},
distributed register automata~\cite{BBR19}, and distributed memory
automata~\cite{BRS21}. Such formalisms, inspired by register
automata~\cite{KF94}, allow identities, control structures and
alternative communication mechanisms; none allowed in population
protocols.

\subparagraph*{Paper organization.}

\Cref{sec:prelims} provides basic definitions and introduces our
model. In \Cref{sec:maj}, we present a protocol that computes the
absolute majority predicate. \Cref{sec:io} establishes the expressive
power of immediate observation protocols. We conclude in
\Cref{sec:conclusion}.

\renewcommand{\c}{\mat{C}}

\section{Preliminaries}\label{sec:prelims}
We write $\N$ and $[a..b]$ to respectively denote sets $\{0, 1, 2,
\ldots\}$ and $\{a, a+1, \ldots, b\}$. The \emph{support} of a
multiset $\vec{m}$ over $E$ is $\act{\vec{m}} \defeq \{e \in E :
\vec{m}(e) > 0\}$ (We use the notation $\act{\vec{m}}$ rather than
$\supp{\vec{m}}$ as we will later refer to ``active states''.) We
write $\N^E$ to denote the set of multisets over $E$ with finite
support. The empty multiset, denoted $\0$, is such that $\0(e) = 0$
for all $e \in E$. Let $\vec{m}, \vec{m}' \in \N^E$. We write $\vec{m}
\leq \vec{m}'$ iff $\vec{m}(e) \leq \vec{m}'(e)$ for all $e \in E$.
We define $\vec{m} + \vec{m}'$ as the multiset such that $(\vec{m} +
\vec{m}')(e) \defeq \vec{m}(e) + \vec{m}'(e)$ for all $e \in E$. The
difference, denoted $\vec{m} - \vec{m}'$, is defined similarly
provided that $\vec{m} \geq \vec{m}'$.

\subsection{Population protocols with unordered data}\label{ssec:udpp}

A \emph{population protocol with unordered data}, over an
infinite domain $\D$ equipped with equality, is a tuple $(Q, \delta,
I, O)$ where
\begin{itemize}
\item $Q$ is a finite set of elements called \emph{states},

\item $\delta \subseteq Q^2 \times \{=, \neq\} \times Q^2$ is the set
  of \emph{transitions},

\item $I \subseteq Q$ is the set of \emph{initial states}, and

\item $O \colon Q \to \{\false, \true\}$ is the \emph{output}
  function.
\end{itemize}

We refer to an element of $\D$ as a \emph{datum} or as a
\emph{color}. We will implicitly assume throughout the paper that
$\delta$ contains $((p, q), \sim, (p, q))$ for all $p, q \in Q$ and
${\sim} \in \{=, \neq\}$.

A \emph{form} $\vec{f}$ is an element from $\N^Q$. We denote the set
of all forms by $\F$. Given $Q' \subseteq
Q$, let $\vec{f}(Q') \defeq \sum_{q \in Q'} \vec{f}(q)$. A
\emph{configuration} is a mapping $\c$ from $\D$ to $\F$ such that
$\supp{\c} \defeq \{d \in \D : \c(d) \neq \0\}$ is finite, and
$\sum_{d \in \D, q \in Q} \c(d)(q) \geq 2$. We often write $\c(d)(q)$
as $\c(d, q)$. Informally, the latter denotes the number of agents
with datum $d$ in state $q$. We extend this notation to subsets of
states: $\c(d, Q') \defeq \c(d)(Q')$. We naturally extend $+$, $-$ and
$\leq$ to $\D \to \F$, \eg\ $\c + \c'$ is such that $(\c + \c')(d)
\defeq \c(d) + \c'(d)$ for all $d \in \D$.

Let $\c$ be a configuration. We define the \emph{active states} 
as the set $\act{\c} \defeq \{q \in Q : \c(d, q) > 0 \text{ for some } d \in
\D\}$. We say that $\c$ is \emph{initial} if $\act{\c} \subseteq
I$. Given $Q' \subseteq Q$, let $\card{Q'}{\c} \defeq \sum_{d \in \D}
\c(d, Q')$ and $|\c| \defeq \card{Q}{\c}$. The \emph{output} of $\c$
is defined by $O(\c) \defeq b$ if $O(q) = b$ for every $q \in \act{\c}$;
and by $O(\c) = \bot$ otherwise. Informally, $O(\c)$ indicates whether
all agents agree on some output $b$.

\begin{example}\label{ex:config}
  Let $\D \defeq \{\tokone, \toktwo, \tokthree, \ldots\}$ and $Q
  \defeq \{p, q\}$. Let $\vec{f} \defeq \multiset{p, p, q}$ and
  $\vec{f}' \defeq \multiset{q}$. Let $\c \defeq \{\tokone \mapsto
  \vec{f}, \toktwo \mapsto \vec{f}', \tokthree \mapsto \vec{0},
  \ldots\}$. We have $\mat{C}(\tokone, p) = 2$, $\mat{C}(\tokone, q) =
  \mat{C}(\toktwo, q) = 1$, $\mat{C}(\toktwo, p) = \mat{C}(\tokthree,
  p) = \mat{C}(\tokthree, q) = 0$ and $\card{\{q\}}{\c} =
  2$. Configuration $\mat{C}$ represents a population of four agents
  carrying an immutable datum and a mutable state:
  $\multiset{(\tokone, p), (\tokone, p), (\tokone, q), (\toktwo,
    q)}$. \qed
\end{example}

For the sake of brevity, given a form $\f$, let $\vec{f}_d \colon \D
\to \F$ be defined by $\vec{f}_d(d) \defeq \vec{f}$ and $\vec{f}_d(d')
\defeq \0$ for every $d' \neq d$. Furthermore, given a state $q$, let
$\q_d \colon \D \to \F$ be defined by $\q_d(d)(q) \defeq 1$ and
$\q_d(d')(q') \defeq 0$ for every $(d', q') \neq (d, q)$.

Let $\c$ be a configuration and let $t = ((p, q), \sim, (p', q')) \in
\delta$. We say that transition $t$ is \emph{enabled} in $\c$ if there
exist $d, e \in \D$ such that $d \sim e$, $\c \geq \p_d + \q_e$. If the
latter holds, then $t$ can be used to obtain the configuration $\c'
\defeq \c - (\p_d + \q_e) + (\vec{p'}_d + \vec{q'}_e)$,  which we denote $\c
\reach{t} \c'$. We write $\c \reach{} \c'$ to denote that $\c
\reach{t} \c'$ holds for some $t \in \delta$. We further define
$\reach{*}$ as the reflexive-transitive closure of $\reach{}$.

\begin{example}
  Let $O(p) \defeq \false$, $O(q) \defeq \true$ and $t \defeq ((p, q),
  =, (q, q))$. Using the notation of \Cref{ex:config} to represent
  configurations, we have:
  \[
  \multiset{(\tokone, p), (\tokone, p), (\tokone, q), (\toktwo, q)}
  \reach{t}
  \multiset{(\tokone, p), (\tokone, q), (\tokone, q), (\toktwo, q)}
  \reach{t}
  \multiset{(\tokone, q), (\tokone, q), (\tokone, q), (\toktwo, q)}.
  \]
  Let $\c$, $\c'$ and $\c''$ denote the three configurations above. We
  have $O(\c) = O(\c') = \bot$ and $O(\c'') = \true$. Moreover,
  transition $t$ is not enabled in $\c''$ as no datum $d \in \D$
  satisfies $\c''(d, p) \geq 1$ and $\c''(d, q) \geq 1$. So, the
  agents have ``converged to a $\true$-consensus''. \qed
\end{example}

An \emph{execution} is an infinite sequence of configurations $\c_0
\c_1 \cdots$ such that $\c_0 \reach{} \c_1 \reach{} \cdots$. We say
that such an execution \emph{converges} to output $b \in \{\false,
\true\}$ if there exists $\tau \in \N$ such that $O(\c_\tau) =
O(\c_{\tau+1}) = \cdots = b$. An execution $\c_0 \c_1 \cdots$ is
\emph{fair} if, for every configuration $\c'$, it is the case that
$|\{i \in \N : \c_i \reach{*} \c'\}| = \infty$ implies $|\{i \in \N :
\c_i = \c'\}| = \infty$. In words, fairness states that if $\c'$ is
reachable infinitely often, then it appears infinitely often along the
execution. Informally, this means that some ``progress'' cannot be
avoided forever.

Let $\Sigma$ be a nonempty finite set. An \emph{input} is some
$\mat{M} \in \N^{\D \times \Sigma}$ with $\sum_{d \in \D, \sigma \in
  \Sigma} \mat{M}(d, \sigma) \geq 2$. An input $\mat{M}$ is
translated, via a bijective \emph{input mapping} $\iota \colon \Sigma
\to I$, into the initial configuration $\iota(\mat{M}) \defeq \sum_{d
  \in \D, \sigma \in \Sigma} \sum_{j=1}^{\mat{M}(d, \sigma)}
\vec{\iota(\sigma)}_d$. We say that a protocol \emph{computes} a predicate
$\varphi$ if, for every input $\mat{M}$, every fair execution starting
in $\iota(\mat{M})$ converges to output $\varphi(\mat{M})$. By abuse
of notation, we sometimes write $\varphi(\c_0)$ for $\varphi(\iota^{-1}(\c_0))$.

\begin{example}
  Let $\D \defeq \{\tokone, \toktwo, \tokthree, \ldots\}$, $\Sigma
  \defeq \{x_1, \ldots, x_4\}$, $I \defeq \{q_1, \ldots, q_4\}$ and
  $\iota(x_i) \defeq q_i$. The input $\mat{M} \defeq \multiset{
  (\tokone, x_1), (\tokone, x_1),
  (\tokone, x_2), (\tokone, x_2), (\tokone, x_2),
  (\tokone, x_4),
  (\toktwo, x_1), (\toktwo, x_3)
  }$
  yields the initial configuration
  $
  \iota(\mat{M}) = 
  \{
  \tokone \mapsto \multiset{q_1, q_1, q_2, q_2, q_2, q_4},
  \toktwo \mapsto \multiset{q_1, q_3},
  \tokthree \mapsto \vec{0},
  \ldots
  \}
  $.
\end{example}

Observe that, as for standard protocols, the set of predicates
computed by population protocols with unordered data is closed under
Boolean operations. Given a protocol that computes $\varphi$, we
obtain a protocol that computes $\lnot \varphi$ by changing the value
of $O(q)$ to $\lnot O(q)$ for all $q \in Q$. Given predicates $\psi_1$
and $\psi_2$, respectively computed by protocols $(Q_1, \delta_1, I_1,
O_1)$ and $(Q_2, \delta_2, I_2, O_2)$, it is easy to obtain a protocol
computing $\psi = \psi_1 \land \psi_2$ by having both protocols run in
parallel. This is achieved by defining $(Q \defeq Q_1 \times Q_2,
\delta, I \defeq I_1\times I_2, O)$ where
\begin{itemize}
\item $\delta$ contains $(((p_1, p_2), (q_1, q_2)), \sim, ((p_1',
  p_2), (q_1', q_2)))$ for every $((p_1, q_1), \sim, (p_1', q_1')) \in
  \delta_1$;

\item $\delta$ contains $(((p_1, p_2), (q_1, q_2)), \sim, ((p_1,
  p_2'), (q_1, q_2')))$ for every $((p_2, q_2), \sim, (p_2', q_2'))
  \in \delta_2$;
  
\item $O(q_1, q_2) = O_1(q_1) \land O_2(q_2)$.
\end{itemize}

\section{A protocol for the majority predicate}\label{sec:maj}
\newcommand{\majphi}{\varphi_{\text{maj}}}

Let $\Sigma \defeq \{x\}$. In this section, we present a protocol for the
absolute majority predicate defined as $\majphi(\mat{M}) \defeq \exists d \in
\D : \mat{M}(d, x) > \sum_{d'\neq d} \mat{M}(d', x)$. Since each input
pair has the form $(d, x)$ with $d \in \D$, we omit the ``dummy
element'' $x$ in the informal presentation of the protocol. Note that for the sake of brevity, we use the term majority instead of absolute majority for the remainder of this paper.

Our protocol is not unlike the classical (sequential) Boyer--Moore
algorithm~\cite{BM91}: we seek to elect a color as the majority
candidate, and then check whether it indeed has the majority. It is
intended to work in stages. In the \emph{pairing stage}, each unpaired
agent seeks to form a pair with an unpaired agent of a distinct
color. For example, if the initial population is
$\multiset{\tokone, \tokone, \tokone, \tokone, \toktwo, \toktwo, \tokthree}$,
then we (non-deterministically) end up with either of these two
pairings:
\begin{center}
  \begin{tabular}{ll}
    \toprule
    \emph{paired agents} & \emph{agents left unpaired} \\
    \midrule
    $\multiset{
      \tokone\link\toktwo,
      \tokone\link\toktwo,
      \tokone\link\tokthree}$ &    
    $\multiset{\tokone}$ \\
    $\multiset{
      \tokone\link\toktwo,
      \toktwo\link\tokthree}$ &    
    $\multiset{\tokone, \tokone, \tokone}$ \\
    \bottomrule
  \end{tabular}
\end{center}

The agents left unpaired must all have the same color $d$,
\eg\ ``$\tokone$'' in the above example. Moreover, if the population
has a majority color, then it must be $d$.

Since the agents are anonymous and have a finite memory, they cannot
actually remember with whom they have been paired. Thus, once a
candidate color has been elected, \eg ``$\tokone$'' in the above
example, there is a \emph{grouping stage}. In the latter, unpaired
agents indicate to agents of their color that they are part of the
candidate majority group. This is done by internally storing the value
``$Y$'', which stands for ``Yes''. Similarly, unpaired agents indicate
to agents of a distinct color that they are part of the candidate
minority group using ``$N$''. Once this is over, the \emph{majority
  stage} takes place using the classical protocol from the
introduction.

Two issues arise from this idealized description. First, the protocol
is intended to work in stages, but they may occur concurrently due to
their distributed nature. For this reason, we add a correction
mechanism:
\begin{itemize}
\item If an unpaired agent of the candidate majority color $d$ finds a
  paired agent of color $d$ (resp.\ $d' \neq d$) with ``$N$''
  (resp.\ ``$Y$''), then it flips it to ``$Y$'' (resp.\ ``$N$'');

\item If an unpaired agent of the candidate majority color $d$ finds a
  paired agent of color $d$ (resp.\ $d' \neq d$) with either ``$n$''
  or ``$y$'', then it flips it to ``$\overline{Y}$''
  (resp.\ ``$\overline{N}$'').
\end{itemize}
The intermediate value $\overline{Y}$ (resp.\ $\overline{N}$) must be
reverted to $Y$ (resp.\ $N$) by finding an agent that has initially played role
$Y$ (resp.\ $N$) and is then reset to its original value.

The second issue has to do with the fact that, in even-size
populations, all agents may get paired. In that case, no unpaired
agent is left to group the agents. To address this, each agent carries
an ``even bit'' to indicate its belief on whether some unpaired agent
remains.

\subsection{States}

The set of states is defined as $Q \defeq \{\false, \true\}^3 \times
\{Y, N, \overline{Y}, \overline{N}, y, n\}$. To ease the reader's
understanding, we manipulate states with four ``macros''. Each macro has a
set of possible values; each state is a combination of values for the
different macros.
\begin{center}
  \begin{tabular}{ccc}
    \begin{tabular}{cccc}
      \toprule 
      \emph{name} & \emph{values for $q\in Q$} & \emph{value for $q \in I$} \\
      \midrule
      $\pair{q}$ & $\{\false, \true\}$ & $\false$ \\
      $\grp{q}$  & $\{\false, \true\}$ & $\true$  \\
      \bottomrule
    \end{tabular}
    & \hspace*{-14pt} &
    \begin{tabular}{cccc}
      \toprule 
      \emph{name} & \emph{values for $q\in Q$} & \emph{value for $q \in I$} \\
      \midrule
      $\even{q}$ & $\{\false, \true\}$ & $\false$ \\
      $\maj{q}$  & $\{Y, N, \overline{Y}, \overline{N}, y, n\}$ & $Y$ \\
      \bottomrule
    \end{tabular}
  \end{tabular}
\end{center}

The input mapping is defined by $\iota(x) \defeq q_I$, where $q_I$ is
the unique state of $I$. Informally, $\pair{q}$ indicates whether the
agent has been paired; $\grp{q}$ indicates whether the agent belongs
to the candidate majority group; $\maj{q}$ is the current value of the
majority computation; and $\even{q}$ is the even bit.

\subsection{Transitions and stages}

We describe the protocol by introducing rules corresponding to
each stage. Note that a rule is a structure on which transitions can be based; therefore, a single rule can yield multiple transitions of the same nature. For convenience, some lemmas are stated before they can
actually be proven, as they require the full set of transitions to be
defined first. Proofs in the appendix take into account the complete
list of transitions.

As the set of transitions for the protocol is lengthy, we present it
using a ``precondition-update'' notation where for any two agents in
state $p, q \in Q$, respectively with colors $d_1, d_2 \in \D$, a
single transition whose preconditions on $p, q$ and $d_1, d_2$ are met
is used. The result of such an interaction is the agent initially in
state $p$ updating its state to $p'$, where $p'$ is identical to $p$
except for the specified macros; and likewise for $q$. To help the readability, the
precondition and update of states $p$ and $q$ are on distinct lines in
the forthcoming tables.

\subsubsection{Pairing stage}

The first rule is used for the pairing stage whose main goal is to
match as many agents as possible with agents of a different color:

\begin{center}
  \begin{tabular}{cccc}
    \toprule
    \emph{rule}
    & \emph{state precondition}
    & \emph{color precondition}
    & \emph{state update} \\

    \midrule\phantomsection\label{maj1}
    \multirow{2}{*}{(1)}
    & $\lnot\pair{p}$
    & \multirow{2}{*}{$d_1 \neq d_2$}
    & $\pair{p'} \land \even{p'}$ \\

    & $\lnot\pair{q}$
    &
    & $\pair{q'} \land \even{q'}$ \\
    \bottomrule
  \end{tabular}
\end{center}
This rule gives rise to the following lemmas concerning the end of the
pairing stage and the nature of unpaired agents, if they exist. For
the remainder of the section, let us fix a fair execution $\c_0
\c_1\cdots$ where $\c_0$ is initial. Moreover, let $P \defeq \{q \in Q
: \pair{q}\}$ and $U \defeq Q \setminus P$.

\begin{restatable}{lemma}{lempair}\label{lem:pair}
  There exists $\tau \in \N$ such that
  $\card{U}{\c_\tau} = \card{U}{\c_{\tau+1}} = \cdots$. Furthermore,
  for every $i \geq \tau$, all unpaired agents of $\c_i$ share the
  same color, \ie\ the set $\{d \in \D : \c_i(d, U) > 0\}$ is either empty or a singleton.
\end{restatable}

Let $\alpha$ denote the minimal threshold $\tau$ given by
\Cref{lem:pair}, which is informally the ``end of the pairing
stage''.

\begin{restatable}{lemma}{lemmajcolor}\label{lem:majcolor}
  Let $i \in \N$. If $\majphi(\c_0)$
  and $d$ is the majority color, then $\c_i(d, U) > 0$.
\end{restatable}

\subsubsection{Grouping stage}

The next set of transitions seeks to correctly set each agent's group,
representing its status in the computation of the majority. An agent
is either part of the candidate majority group ($\true$), or part of
the candidate minority group ($\false$). Note that this group (and its
related majority computing value) are irrelevant if there are no
unpaired agents in $\c_\alpha$; this special case is handled using the
even bit, which is ignored for now.

\begin{center}
  \begin{tabular}{cccc}
    \toprule
    \emph{rule}
    & \emph{state precondition}
    & \emph{color precondition}
    & \emph{state update} \\

    \midrule\phantomsection\label{maj2}    
    \multirow{2}{*}{(2)}
    & \makebox[130pt][l]{$\lnot\pair{p}$}
    & \multirow{2}{*}{$d_1 \neq d_2$}
    & \none \\
    & \makebox[130pt][l]{$\phantom{\lnot}\pair{q} \land \phantom{\lnot}\grp{q} \land \maj{q} = Y$}
    &
    & \makebox[100pt][r]{$\lnot\grp{q'} \land \maj{q'} = N$} \\

    \midrule\phantomsection\label{maj3}    
    \multirow{2}{*}{(3)}
    & \makebox[130pt][l]{$\lnot\pair{p}$}
    & \multirow{2}{*}{$d_1 = d_2$}
    & \none \\
    & \makebox[130pt][l]{$\phantom{\lnot}\pair{q} \land \lnot\grp{q} \land \maj{q} = N$}
    &
    & \makebox[100pt][r]{$\phantom{\lnot}\grp{q'} \land \maj{q'} = Y$} \\

    \midrule\phantomsection\label{maj4}    
    \multirow{2}{*}{(4)}
    & \makebox[130pt][l]{$\lnot\pair{p}$}
    & \multirow{2}{*}{$d_1 \neq d_2$}
    & \none \\
    & \makebox[130pt][l]{$\phantom{\lnot}\pair{q} \land \phantom{\lnot}\grp{q} \land \maj{q} \in \{y, n\}$}
    &
    & \makebox[100pt][r]{$\maj{q'} = \overline{N}$} \\

    \midrule\phantomsection\label{maj5}    
    \multirow{2}{*}{(5)}
    & \makebox[130pt][l]{$\lnot\pair{p}$}
    & \multirow{2}{*}{$d_1 = d_2$}
    & \none \\
    & \makebox[130pt][l]{$\phantom{\lnot}\pair{q} \land \lnot\grp{q} \land \maj{q} \in \{y, n\}$}
    &
    & \makebox[100pt][r]{$\maj{q'} = \overline{Y}$} \\
    \bottomrule
  \end{tabular}
\end{center}

The forthcoming rules below are part of a two-rule combination whose
aim is to rectify an error in grouping assignments. It allows agents
who engaged in the computation within the candidate minority group
(resp.\ majority group) who encountered a currently valid majority
candidate of their color (resp.\ a different color) to reset their
value to $Y$ (resp.\ $N$) and their group to $\true$ (resp.\ $\false$)
by finding another agent, also engaged, to do the same. This, along
with the rules described in the next subsection, ensures that the
invariant below holds.

Let $Q_a \defeq \{q \in Q : \maj{q} = a\}$, $Q_M \defeq \{q \in Q :
\grp{q}\}$ and $Q_m \defeq Q \setminus Q_M$.
  
\begin{restatable}{lemma}{leminvariant}\label{lem:invariant}
  For every $i \in \N$, it is the
  case that $\card{Q_Y}{\c_i} - \card{Q_N}{\c_i} = \card{Q_M}{\c_i} -
  \card{Q_m}{\c_i}$.
\end{restatable}

\begin{center}
  \begin{tabular}{cccc}
    \toprule
    \emph{rule}
    & \emph{state precondition}
    & \emph{color precondition}
    & \emph{state update} \\

    \midrule\phantomsection\label{maj6}
    \multirow{2}{*}{(6)}
    & \makebox[100pt][l]{$\phantom{\lnot}\grp{p} \land \maj{p} = \overline{N}$}
    & \multirow{2}{*}{\none}
    & \makebox[100pt][r]{$\lnot\grp{p'} \land \maj{p'} = N$} \\
    & \makebox[100pt][l]{$\lnot\grp{q} \land \maj{q} \in \{y, n\}$}
    &
    & \makebox[100pt][r]{$\maj{q'} = N$} \\

    \midrule\phantomsection\label{maj7}
    \multirow{2}{*}{(7)}
    & \makebox[100pt][l]{$\lnot\grp{p} \land \maj{p} = \overline{Y}$}
    & \multirow{2}{*}{\none}
    & \makebox[100pt][r]{$\phantom{\lnot}\grp{p'} \land \maj{p'} = Y$} \\
    & \makebox[100pt][l]{$\phantom{\lnot}\grp{q} \land \maj{q} \in \{y, n\}$}
    &
    & \makebox[100pt][r]{$\maj{q'} = Y$} \\

    \midrule\phantomsection\label{maj8}
    \multirow{2}{*}{(8)}
    & \makebox[100pt][l]{$\phantom{\lnot}\grp{p} \land \maj{p} = \overline{N}$}
    & \multirow{2}{*}{\none}
    & \makebox[100pt][r]{$\lnot\grp{p'} \land \maj{p'} = n$} \\
    & \makebox[100pt][l]{$\lnot\grp{q} \land \maj{q} = \overline{Y}$}
    &
    & \makebox[100pt][r]{$\phantom{\lnot}\grp{q'} \land \maj{q'} = n$} \\
    \bottomrule
  \end{tabular}
\end{center}

We give the following example to help illustrate the necessity of the intermediate states $\overline{Y}, \overline{N}$ in the context of the suggested protocol.

\begin{example}\label{ex:intermediate}
  Let us first consider a possible execution from the initial population $\multiset{\tokone, \tokone, \toktwo, \toktwo, \tokthree}$, for which there is no majority datum. Observe that since the number of agents is odd, in any execution, there will be a datum with an unpaired agent after the pairing stage. Assume, for the sake of our demonstration, that this datum is blue ($\tokthree$). Entering the grouping stage, this blue agent will eventually let the other agents know that they are not part of the majority candidate group and, at some point, rule~\majrule{11} will occur, leading to all agents permanently with $\mathrm{maj}(q) \in \{N, n\}$. This is summarized in these three snapshots (where the even bit is omitted for the sake of clarity):

\newcommand{\cmark}{\text{\scriptsize\faCheck}}
\newcommand{\xmark}{\text{\scriptsize\faTimes}}
\begin{center}
  \begin{tabular}{cccc}
    \toprule

    input & $\mathrm{pair}$ & 
    $\mathrm{grp}$ & $\mathrm{maj}$ \\

    \midrule

    $\tokone$ & \xmark & \cmark & $Y$ \\
    $\tokone$ & \xmark & \cmark & $Y$ \\
    $\toktwo$ & \xmark & \cmark & $Y$ \\
    $\toktwo$ & \xmark & \cmark & $Y$ \\
    $\tokthree$ & \xmark & \cmark & $Y$ \\

    \bottomrule
  \end{tabular}
  $\xrightarrow{*}$
    \begin{tabular}{cccc}
      \toprule

      input & $\mathrm{pair}$ & 
      $\mathrm{grp}$ & $\mathrm{maj}$ \\

      \midrule

      $\tokone$ & \cmark & \xmark & $N$ \\
      $\tokone$ & \cmark & \xmark & $N$ \\
      $\toktwo$ & \cmark & \xmark & $N$ \\
      $\toktwo$ & \cmark & \xmark & $N$ \\
      $\tokthree$ & \xmark & \cmark & $Y$ \\

      \bottomrule
    \end{tabular}
    $\xrightarrow{}$
    \begin{tabular}{cccc}
      \toprule

      input & $\mathrm{pair}$ & 
      $\mathrm{grp}$ & $\mathrm{maj}$ \\

      \midrule

      $\tokone$ & \cmark & \xmark & $N$ \\
      $\tokone$ & \cmark & \xmark & $n$ \\
      $\toktwo$ & \cmark & \xmark & $N$ \\
      $\toktwo$ & \cmark & \xmark & $N$ \\
      $\tokthree$ & \xmark & \cmark & $n$ \\

      \bottomrule
    \end{tabular}
\end{center}

  Now, consider a population where a majority datum does indeed exist: $\multiset{\tokone, \tokone, \tokone, \tokone, \toktwo, \toktwo, \tokthree}$. Note that this population is strictly greater than the previous population. Therefore, we can promptly obtain a configuration similar to the one described above, where two more agents of datum red ($\tokone$) have yet to participate in the computation. Since the computation must output $\true$, the consensus on $\{N, n\}$ initiated by the blue agent has to be reverted.

  In this case, after the final pairing is done via an interaction between the blue agent ($\tokthree$) and one of the newly introduced red agents ($\tokone$), the error handling first works through rule~\majrule{4} or~\majrule{5}: the unpaired red agent ($\tokone$) notifies the blue agent ($\tokthree$) that its group is incorrect by setting its computation value to $\overline{N}$ and similarly, it notifies all red agents ($\tokone$) who had previously participated in the (now incorrect) majority stage to switch their computation value to $\overline{Y}$. This is summarized in these three snapshots:
  \begin{center}

    $\cdots$
      \begin{tabular}{cccc}
        \toprule

        input & $\mathrm{pair}$ & 
        $\mathrm{grp}$ & $\mathrm{maj}$ \\

        \midrule

        $\tokone$ & \cmark & \xmark & $N$ \\
        $\tokone$ & \cmark & \xmark & $n$ \\
        $\toktwo$ & \cmark & \xmark & $N$ \\
        $\toktwo$ & \cmark & \xmark & $N$ \\
        $\tokthree$ & \cmark & \cmark & $n$ \\
        \midrule
        $\tokone$ & \cmark & \cmark & $Y$ \\
        $\tokone$ & \xmark & \cmark & $Y$ \\

        \bottomrule
      \end{tabular}
      $\xrightarrow{*}$
      \begin{tabular}{cccc}
        \toprule

        input & $\mathrm{pair}$ & 
        $\mathrm{grp}$ & $\mathrm{maj}$ \\

        \midrule

        $\tokone$ & \cmark & \cmark & $Y$ \\
        $\tokone$ & \cmark & \xmark & $\overline{Y}$ \\
        $\toktwo$ & \cmark & \xmark & $N$ \\
        $\toktwo$ & \cmark & \xmark & $N$ \\
        $\tokthree$ & \cmark & \cmark & $\overline{N}$ \\
        \midrule
        $\tokone$ & \cmark & \cmark & $Y$ \\
        $\tokone$ & \xmark & \cmark & $Y$ \\

        \bottomrule
      \end{tabular}
      $\xrightarrow{}$
      \begin{tabular}{cccc}
        \toprule

        input & $\mathrm{pair}$ & 
        $\mathrm{grp}$ & $\mathrm{maj}$ \\

        \midrule

        $\tokone$ & \cmark & \cmark & $Y$ \\
        $\tokone$ & \cmark & \cmark & $n$ \\
        $\toktwo$ & \cmark & \xmark & $N$ \\
        $\toktwo$ & \cmark & \xmark & $N$ \\
        $\tokthree$ & \cmark & \xmark & $n$ \\
        \midrule
        $\tokone$ & \cmark & \cmark & $Y$ \\
        $\tokone$ & \xmark & \cmark & $Y$ \\

        \bottomrule
      \end{tabular}
  \end{center}
  This inevitably leads each incorrectly grouped agent to rectify its group bit as well as its computation value, accordingly, through rules~\majrule{6},~\majrule{7} or~\majrule{8}. We then have a configuration for which the grouping stage is over and where either the majority stage is not yet initiated, or it has been correctly initiated with the right majority candidate. \qed
\end{example}

The following lemmas show that the grouping stage eventually ends if
there are unpaired agents in $\c_\alpha$. Moreover, they show that the
majority candidate color $d$ eventually propagates the majority group
to agents of color $d$, and the minority group to agents of color $d'
\neq d$.

\begin{restatable}{lemma}{lempartner}\label{lem:partner}
  Let $E$ be the set of states engaged
  in the majority computation, \ie\ $E \defeq \{q\in Q : \maj{q} \in
  \{y, n, \overline{Y}, \overline{N}\}\}$. Let $E_M \defeq E \cap Q_M$
  and $E_m \defeq E \cap Q_m$. For every $i \in \N$, the following
  holds: $\card{E_M}{\c_i} = \card{E_m}{\c_i}$.
\end{restatable}

\begin{restatable}{lemma}{lemgroup}\label{lem:group}
  Let $d \in \D$. If $\c_\alpha(d, U) > 0$, then there exists some $\tau \geq \alpha$
  such that, for all $i \geq \tau$, $d' \in \D$ and $q \in
  \act{\c_i(d')}$, the following holds: $\grp{q} = (d' = d)$.
\end{restatable}

\subsubsection{Majority stage}

The last set of transitions emulates a standard population protocol
for the majority predicate. Populations of even size without a majority give rise to a case requiring careful handling. Indeed, for such a population the pairing stage may leave no unmatched agent. Therefore, we give the following rules to fix this specific issue.

\begin{center}
  \begin{tabular}{cccc}
    \toprule
    \emph{rule}
    & \emph{state precondition}
    & \emph{color precondition}
    & \emph{state update}\\
    \midrule\phantomsection\label{maj9}

    \multirow{2}{*}{(9)}
    & $\lnot\pair{p}\phantom{\land \lnot\even{q}}$
    & \multirow{2}{*}{\none}
    & \none \\
    & $\phantom{\lnot\pair{p} \land \lnot}\even{q}$
    &
    & $\lnot\even{q'}$ \\

    \midrule\phantomsection\label{maj10}
    \multirow{2}{*}{(10)}
    & $\phantom{\lnot}\pair{p} \land \phantom{\lnot}\even{p}$
    & \multirow{2}{*}{\none}
    & \none \\
    & $\phantom{\lnot}\pair{q} \land \lnot\even{q}$
    &
    & $\phantom{\lnot}\even{q'}$ \\
    \bottomrule
  \end{tabular}
\end{center}

\begin{restatable}{lemma}{lemeven}\label{lem:even}
  There exists $\tau \geq \alpha$ such that
  for every $i \geq \tau$ and $q \in \act{\c_i}$, it is the case that
  $\even{q}$ holds iff $\card{U}{\c_i} = 0$.
\end{restatable}

For other populations, a unique candidate color for the majority exists following the pairing stage. For the predicate to be $\true$, this candidate must have more agents
than all of the other colors combined. This is validated (or invalidated) through the
following rules.

\begin{center}
  \begin{tabular}{ccc}
    \begin{tabular}{cccc}
      \toprule
      \emph{rule}
      & \emph{state pre.}
      & \hspace*{-5pt}\emph{col.\ pre.}\hspace*{-5pt}
      & \emph{state update} \\
      
      \midrule\phantomsection\label{maj11}
      \multirow{2}{*}{(11)}
      & \makebox[50pt][l]{$\maj{p} = Y$}
      & \multirow{2}{*}{\none}
      & \makebox[50pt][l]{$\maj{p'} = n$} \\
      & \makebox[50pt][l]{$\maj{q} = N$}
      &
      & \makebox[50pt][l]{$\maj{q'} = n$} \\

      \midrule\phantomsection\label{maj12}
      \multirow{2}{*}{(12)}
      & \makebox[50pt][l]{$\maj{p} = Y$}
      & \multirow{2}{*}{\none}
      & \none \\
      & \makebox[50pt][l]{$\maj{q} = n$}
      &
      & \makebox[50pt][l]{$\maj{q'} = y$} \\

      \bottomrule
    \end{tabular}
    & \hspace*{-12pt} &
    \begin{tabular}{cccc}
      \toprule
      \emph{rule}
      & \emph{state pre.}
      & \hspace*{-5pt}\emph{col.\ pre.}\hspace*{-5pt}
      & \emph{state update} \\

      \midrule\phantomsection\label{maj13}
      \multirow{2}{*}{(13)}
      & \makebox[50pt][l]{$\maj{p} = N$}
      & \multirow{2}{*}{\none}
      & \none \\
      & \makebox[50pt][l]{$\maj{q} = y$}
      &
      & \makebox[50pt][l]{$\maj{q'} = n$} \\
      
      \midrule\phantomsection\label{maj14}
      \multirow{2}{*}{(14)}
      & \makebox[50pt][l]{$\maj{p} = n$}
      & \multirow{2}{*}{\none}
      & \none \\
      & \makebox[50pt][l]{$\maj{q} = y$}
      &
      & \makebox[50pt][l]{$\maj{q'} = n$} \\
    \bottomrule
    \end{tabular}
  \end{tabular}
\end{center}

\begin{restatable}{lemma}{lemconvtrue}\label{lem:convtrue}
  If $\majphi(\c_0)$ holds, then there
  exists $\tau \geq \alpha$ such that for every $i \geq \tau$ and $q
  \in\act{\c_i}$, it is the case that $\maj{q} \in \{Y, y\}$ and
  $\lnot\even{q}$ hold.
\end{restatable}

\begin{restatable}{lemma}{lemconvfalse}\label{lem:convfalse}
  If $\neg \majphi(\c_0)$ holds,
  then there exists $\tau \geq \alpha$ such that either:
  \begin{itemize}
  \item $\even{q}$ holds for every $i \geq \tau$ and $q \in
    \act{\c_i}$; or

  \item $\maj{q} \in \{ N, n\}$ holds for every $i \geq \tau$ and $q
    \in \act{\c_i}$.
  \end{itemize}
\end{restatable}
\medskip

We define the output of a given state $q \in Q$ as $O(q) \defeq
(\maj{q} \in \{Y, y\} \land \lnot\even{q})$. The correctness of the
protocol follows immediately from \Cref{lem:convtrue,lem:convfalse}:

\begin{restatable}{corollary}{thmmaj}\label{thm:maj}
  There exists $\tau \in \N$ such that $O(\c_\tau) = O(\c_{\tau+1}) =
  \cdots = \majphi(\c_0)$.
\end{restatable}

\section{Immediate observation protocols}\label{sec:io}
We say that a population protocol is \emph{immediate observation (IO)} if
each of its transitions has the form $((p, q), \sim, (p, q'))$,
\ie\ only one agent can update its state by ``observing'' the other
agent. There is no restriction on $\sim$, but one can also imagine the
datum to be observed.

In this section, we characterize the expressive power of immediate
observation protocols. First, we establish properties of IO protocols
regarding truncations, thereby allowing us to prove that the majority
predicate is not computable. Then, we show that IO protocols do not
compute more than interval predicates. Finally, we show that every
interval predicate can be computed by an IO protocol. Before
proceeding, let us define interval predicates.

Let $\dot\exists d_1, d_2, \ldots, d_n$ denote a disjoint existential
quantification, \ie\ it indicates that $d_i \neq d_j$ for all $i, j
\in [1..n]$ such that $i \neq j$. A \emph{simple interval predicate},
interpreted over inputs from $\N^{\D \times \Sigma}$, where $\Sigma = \{x_1, \ldots, x_m\}$, is a predicate of the form
\begin{align}
  \psi(\mat{M}) =
  \dot\exists d_1, d_2, \ldots, d_n \in \D :
  \bigwedge_{i = 1}^n \bigwedge_{j = 1}^m \mat{M}(d_i, x_j) \in T(i, j),\label{eq:int:pred}
\end{align}
where $m, n \in \N_{> 0}$, each $T(i, j) \subseteq \N$ is a nonempty
interval, and for every $i \in [1..n]$, there exists $j \in[1..m]$ such
that $0 \notin T(i, j)$. An \emph{interval predicate} is a Boolean
combination of simple interval predicates.

\subsection{State and form truncations}\label{sec:maj:io}
\newcommand{\embeds}{\sqsubseteq}
\newcommand{\U}{\mathcal{U}}
\renewcommand{\S}{\mathcal{S}}

Given configurations $\c, \c'$, we write $\c \embeds \c'$ if there
exists an injection $\rho \colon \D \to \D$ such that $\c(d) \leq
\c'(\rho(d))$ for every $d \in \D$. We write $\c \equiv \c'$ if $\c
\embeds \c'$ and $\c' \embeds \c$. We say that a subset of configurations
$X$ is \emph{upward closed} if $\c \in X$ and $\c \embeds \c'$ implies
$\c' \in X$. We say that a set $B$ is a \emph{basis} of an upward
closed set $X$ if $X = \{\c' : \c \embeds \c' \text{ for some } \c \in
B\}$.

A configuration $\c$ is said \emph{unstable} if either $O(\c) = \bot$
or there exists $\c'$ such that $\c \reach{*} \c'$ with $O(\c) \neq
O(\c')$. Let $\U$ denote the set of unstable configurations, and let
$\S_b \defeq \{\c : \c \not\in \U, O(\c) = b\}$ denote the set of
stable configurations with output $b$. As in the case of standard
protocols (without data)~\cite{AADFP06}, it is simple to see that $\U$
is upward closed. Moreover, since $\embeds$ is a well-quasi-order, it
follows that $\U$ has a finite basis.

This allows us to extend the notion of truncations
from~\cite{AADFP06}. A \emph{state truncation} to $k \geq 1$ of some
form $\f$, denoted by $\tau_k(\f)$, is the form such that
$\tau_k(\f)(q) \defeq \min(\f(q), k)$ for all $q \in Q$. The concept
of state truncations is also extended to configurations: $\tau_k(\c)$
is the configuration such that $\tau_k(\c)(d) \defeq \tau_k(\c(d))$
for all $d \in \D$. From a sufficiently large threshold, the stability
and output of a configuration remain unchanged under state
truncations:
\begin{restatable}{lemma}{lemtruncinu}\label{lem:truncinu}
  Let $\psi$ be a predicate computed by a population protocol with unordered data. Let $S_b$ be the set of stable configurations with output $b$ of the protocol.
  There exists $k \geq 1$ such that,
  for all $b \in \{0, 1\}$, we have $\c \in \S_b$ iff $\tau_k(\c) \in
  \S_b$.
\end{restatable}

\begin{restatable}{lemma}{lemstatelimit}\label{lem:statelimit}
  Let $\psi$ be a predicate computed by an immediate observation protocol with unordered data. Let $k$ be the threshold given by \Cref{lem:truncinu} and $\c_0$ be an initial configuration such that for some $d\in\D$ and some $q\in I$, $\c_0(d, q) \geq \ell = |Q|\cdot (k-1) +1$. Then, $\psi(\c_0 + \vec{q}_d) = \psi(\c_0)$.
\end{restatable}

Given a configuration $\c$ and a form $\f$, let $\#_{\f}(\c) \defeq
|\{d \in \D : \c(d) = \f\}|$. Due to the nature of immediate
observation protocols, it is always possible to take a form $\vec{f}$
of color $d$ present in a configuration $\c$, duplicate
$\vec{f}$ with a fresh color $d'$, and have the latter mimic the
behaviour of the former.

\begin{restatable}{lemma}{lemnewcolor}\label{lem:newcolor}
  Let $\c$ and $\c'$ be configurations
  such that $\c \reach{*} \c'$. For every $d \in \supp{\c}$ and $d'
  \in \D \setminus \supp{\c}$, it is the case that $\c + (\c(d))_{d'}
  \reach{*} \c' + (\c'(d))_{d'}$.
\end{restatable}

Combined with the fact that $\U$ has a finite basis, this allows to
show that from some threshold, duplicating forms with fresh colors
does not change the output of the population.

\begin{restatable}{lemma}{lemformthreshold}\label{lem:formthreshold}
  Let $\psi$ be a predicate
  computed by a population protocol with unordered data. Let $\f$ be a
  form with $\act{\f} \subseteq I$. There exists $h(\f) \in \N$ such
  that, for all initial configuration $\c_0$ and $d \in \D \setminus
  \supp{\c_0}$ with $\#_{\f}(\c_0) \geq h(\f)$, it is the case that
  $\psi(\c_0 + \vec{f}_d) = \psi(\c_0)$.
\end{restatable}

The \emph{form truncation} of a configuration $\c$, denoted
$\sigma(\c)$, is an (arbitrary) configuration such that $\sigma(\c)
\sqsubseteq \c$ and $\#_{\f}(\sigma(\c)) = \min(\#_{\f}(\c), h(\f))$
for every form $\f$, where $h(\f)$ is given by \Cref{lem:formthreshold}. By \Cref{lem:formthreshold}, $\psi(\c_0)$ holds
iff $\psi(\sigma(\c_0))$ holds. Moreover,
\Cref{lem:formthreshold} allows us to show that IO protocols are less
expressive than the general model.

\begin{proposition}
  No IO population protocol computes the majority predicate
  $\varphi_\text{maj}$.
\end{proposition}

\begin{proof}
  For the sake of contradiction, suppose that some IO protocol
  computes $\varphi_\text{maj}$. Let $q_I$ be the unique initial state
  and let $\f \defeq \multiset{q_I}$. Let $h(\f)$ be given by \Cref{lem:formthreshold}. Let $\c_0$ be an initial
  configuration such that
  \begin{itemize}
  \item $\c_0(d) = \sum_{i=1}^{h(\f) + 1} \f$ holds for a unique
    datum $d \in \D$, and
    
  \item $\c_0(d') = \f$ holds for exactly $h(\f)$ other data $d' \in
    \{d_1, d_2, \ldots, d_{h(\f)}\}$.
  \end{itemize}
  
  We have $\varphi_{\text{maj}}(\c_0) = \true$, since $d$ has $h(\f) +
  1$ agents in a population of $2 \cdot h(\f) + 1$ agents. Let $\c_0'$
  be the initial configuration obtained from $\c_0$ by adding a datum
  $d^* \notin \supp{\c_0}$ such that $\c_0'(d^*) = \f$. By
  \Cref{lem:formthreshold}, $\varphi_{\text{maj}}(\c_0') =
  \varphi_{\text{maj}}(\c_0) = \true$. However, datum $d$ no longer
  has a majority in $\c_0'$, which is a contradiction.
\end{proof}

\subsection{Predicates computed by IO protocols are interval predicates}\label{sec:comp}
\begin{theorem}
  Let $(Q, \delta, I, O)$ be an immediate observation protocol with
  unordered data that computes a predicate $\psi$. The predicate
  $\psi$ can be expressed as an interval predicate.
\end{theorem}

\begin{proof}
  We will express $\psi$ as a finite Boolean combination of simple interval
  predicates.
  
  Let $h$ be the mapping given by \Cref{lem:formthreshold}. Let $\mathbb{T} \defeq \{\c : \psi(\c) = \true \}$ and $\mathbb{T}_1
  \defeq \{\c \in \mathbb{T} : \c(d, q) \leq \ell \text{ for all } d \in
  \D, q \in Q\}$. From \Cref{lem:statelimit}, we learn that state
  truncations to $\ell$ on an initial configuration do not alter the consensus reached by the protocol. So,
  $\psi(\c)$ holds iff $\bigvee_{\c'\in \mathbb{T}_1} \tau_{\ell}(\c) =
  \c'$ holds. Let $\mathbb{T}_2 \defeq \{\c \in \mathbb{T}_1 :
  \#_{\f}(\c) \leq h(\f) \text{ for all } \f \in \F\}$. It follows
  from \Cref{lem:formthreshold} that $\psi(\c)$ holds iff
  $\bigvee_{\c'\in \mathbb{T}_2} \sigma(\tau_\ell(\c)) = \c'$ holds.

  The latter is an infinite disjunction. Let us make it
  finite. Observe that if $\c \reach{*} \c'$ and $\overline{\c} \equiv
  \c$ hold, then there exists $\overline{\c'} \equiv \c'$ such that
  $\overline{\c} \reach{*} \overline{\c'}$. Moreover, note that
  equivalent configurations have the same output as they share the
  same active states. Indeed, $\c \equiv \overline{\c}$ iff
  $\bigwedge_{\f\in\F} \#_{\f}(\c) = \#_{\f}(\overline{\c})$. Hence,
  for every initial configuration $\c \equiv \overline{\c}$, we have
  $\psi(\c) = \psi(\overline{\c})$. Let $\mathbb{T}_2 \slash {\equiv}$
  be the set of all equivalence classes of $\equiv$ on $\mathbb{T}_2$,
  and let $\mathbb{T}_3$ be a set that contains one representative
  configuration per equivalence class of $\mathbb{T}_2 \slash
  {\equiv}$. It is readily seen that $\psi(\c)$ holds iff
  $\bigvee_{\c'\in \mathbb{T}_3} \sigma(\tau_\ell(\c)) \equiv \c'$ holds.

  Let us argue that $\mathbb{T}_3$ is finite. Let $\F_\ell \defeq \{\f
  \neq \0 : \f(q) \leq \ell \text{ for all } q \in Q\}$. For every
  configuration $\c \in \mathbb{T}_1$, each form $\f$ with
  $\#_{\f}(\c) > 0$ belongs to $\F_\ell$. As $\mathbb{T}_2 \subseteq
  \mathbb{T}_1$, this also holds for configurations of
  $\mathbb{T}_2$. Given $\c \in \mathbb{T}_2$, we have $\#_{\f}(\c)
  \leq h(\f)$ for all $\f \in \F_\ell$, and $\#_{\f}(\c) = 0$ for all $\f
  \notin \F_\ell$. Thus, as $\F_\ell$ is finite, we conclude that
  $\mathbb{T}_3$ is finite.

  Let us now exploit our observations to express $\psi$ as an interval
  predicate. Let us fix some $\c' \in \mathbb{T}_3$. It suffices to
  explain how to express ``$\sigma(\tau_\ell(\c)) \equiv \c'$''. Indeed,
  as $\mathbb{T}_3$ is finite, we can conclude by taking the finite
  disjunction $\bigvee_{\c'\in \mathbb{T}_3} \sigma(\tau_\ell(\c)) \equiv
  \c'$.

  For every form $\f \in \F_\ell$, let $\mathrm{lt}(\f) \defeq \{q \in Q
  : \f(q) < \ell\}$, $\mathrm{eq}(\f) \defeq \{q \in Q : \f(q) = \ell\}$ and
  \[
  \varphi_{\f, d}(\c) \defeq
  \bigwedge_{\mathclap{q \in \mathrm{lt}(\f)}} (\c(d, q) = \f(q)) \land
  \bigwedge_{\mathclap{q \in \mathrm{eq}(\f)}} (\c(d, q) \geq \f(q)).
  \]
  Observe that $\varphi_{\f, d}(\c)$ holds iff $\tau_\ell(\c)(d) = \f$.
  
  For every $\f \in \F_\ell$ such that $\#_{\f}(\c') < h(\f)$, we define
  this formula, where $n \defeq \#_{\f}(\c')$:
  \begin{align*}
    \psi_{\f}(\c) \defeq \dot\exists d_1, d_2, \ldots, d_n \in\D : \bigwedge_{i=1}^n \varphi_{\f, d_i}(\c) \land
    \lnot\dot\exists d_1, d_2, \ldots, d_{n+1} \in \D : \bigwedge_{i=1}^{n+1} \varphi_{\f, d_i}(\c).
  \end{align*}
  For every $\f \in \F_\ell$ such that $\#_{\f}(\c') = h(\f)$, we define
  this formula, where $n \defeq \#_{\f}(\c')$:
  \begin{align*}
    \psi_{\f}(\c) \defeq \dot\exists d_1, d_2, \ldots, d_n \in\D : & \bigwedge_{i=1}^n \varphi_{\f, d_i}(\c).
  \end{align*}
  Observe that $\psi_{\f}$ is either a simple interval predicate or a
  Boolean combination of two simple interval predicates. Note that
  $\psi_{\f}(\c)$ holds iff $\#_{\f}(\sigma(\tau_\ell(\c))) =
  \#_{\f}(\c')$ holds. This means that $\bigwedge_{\f \in \F_\ell}
  \psi_{\f}(\c)$ holds iff $\sigma(\tau_\ell(\c)) \equiv \c'$ holds, and
  hence we are done.
\end{proof}

\subsection{An IO protocol for simple interval predicates}\label{sec:interval}
As Boolean combinations can be implemented (see end of
\Cref{ssec:udpp}), it suffices to describe a protocol for a simple
interval predicate of the form~\eqref{eq:int:pred}. We refer to each
$i \in [1..n]$ as a \emph{role}. In the forthcoming set of states $Q$,
we associate to each $q \in Q$ an element $\init{q} \in [1..m]$. Each agent's element is set through the input; \eg~an agent mapped from symbol $(\tokone, x_1)$ is initially in a state $q$ such that $\init{q} = 1$. Let
$Q_j \defeq \{q \in Q : \init{q} = j\}$. For any two configurations $\c_a$ and $\c_b$ of an execution, any datum $d\in\D$ and any element $j\in[1..m]$, the invariant $\c_a(d, Q_j) = \c_b(d, Q_j)$ holds. We say that $d \in \D$
\emph{matches} role $i$ in configuration $\c$ if $\c(d, Q_j) \in T(i,
j)$ holds for all $j \in [1..m]$. Let $r \defeq \max(r_1, \ldots, r_n)
+ 1$, where
\[
r_i \defeq \max(\{\min T(i, j) : j \in [1..m]\} \cup \{\max T(i, j) :
j \in [1..m], \sup T(i, j) < \infty\}).
\]
Agents will not need to count beyond value $r$ to decide whether a
role is matched.

As for the majority protocol, our simple interval protocol works in
stages, each one being necessary to ensure properties and invariants
for the subsequent stages. In the \emph{election stage}, a unique
controller for the population and a single leader per datum of the
support are selected; the former seeks to distribute a set of roles to the latter.

All agents contribute to the tallying of their immutable element $j$
through the \emph{counting stage}. This is done using the ``tower
method'' described in~\cite{AAER07}, whereby two agents of the same
datum, element \emph{and} value meet and allow one of the two agents
to increment its value. The maximal value computed in that manner is
subsequently communicated to the (unique) datum leader.

Once the leaders carry correct counts for each element of their
respective datum, they undertake roles that they match in the
\emph{distribution stage}. These roles can be swapped for other roles
(as long as requirements are met) through a process of interrogating
the controller. The controller is constantly notified of selected
roles and updates its list of tasks accordingly.

If a fully assigned task list is obtained by the controller, it
spreads a $\true$-output throughout the population in what we call the
\emph{output propagation stage}. If that is not possible, leaders are
in a consistent state of trial-and-error for their role assignments,
ultimately failing to completely fill the task list, leaving the
controller free to propagate its $\false$-output.

\newcommand{\cmark}{\text{\scriptsize\faCheck}}
\newcommand{\xmark}{\text{\scriptsize\faTimes}}

\begin{example}\label{ex:int}
  Consider $n = m = 2$ with $T(1, 1) \defeq [2..\infty)$,
    $T(1, 2) \defeq [0..4]$, $T(2, 1) \defeq \N$, $T(2, 2) \defeq
    [1..\infty)$. Let $\mat{M} \defeq \multiset{(\tokone, x_1),
        (\tokone, x_1), (\tokone, x_1), (\tokone, x_2), (\toktwo,
        x_1), (\tokthree, x_2)}$. Note that $r = 5$. Datum ``\tokone''
      could match roles~1 and~2, ``\toktwo'' cannot match any role, and
      ``\tokthree\!\!'' could match role~2.

  After executing the protocol for a while, we may end up with the
  configuration illustrated in the table below. The third, fourth,
  fifth and sixth agents contain the correct value for their datum and
  element: $\mat{M}(\tokone, x_1) = 3$ and $\mat{M}(\tokone, x_2) =
  \mat{M}(\toktwo, x_1) = \mat{M}(\tokthree, x_2) = 1$. The second
  agent has been elected controller. The last three agents have been elected their respective datum's
  leader and have collected the correct counts for each element. Either the $\tokone$-leader or the $\tokthree$-leader (possibly both) has notified the controller that they play role $2$.
    
  \begin{center}
      \begin{tabular}{ccccccc}
        \toprule
        
        input & $\mathrm{val}$ & 
        $\mathrm{lead}$ & $\mathrm{ctrl}$ & 
        $\mathrm{role}$ & $\mathrm{count}$ of $[\# x_1, \# x_2]$ & 
        $\mathrm{task}$ list for [role $1$, role $2$] \\
        
        \midrule
        
        $(\tokone, x_1)$ & $1$ & & & & & \\
        $(\tokone, x_1)$ & $2$ & & \cmark & & & $[\xmark, \cmark]$ \\
        $(\tokone, x_1)$ & $\bf 3$ & & & & & \\
        $(\tokone, x_2)$ & $\bf 1$ & \cmark   & & $2$ & $[3, 1]$ & \\
        $(\toktwo, x_1)$ & $\bf 1$ & \cmark   & &     & $[1, 0]$ & \\
        $(\tokthree, x_2)$ & $\bf 1$ & \cmark & & $2$ & $[0, 1]$ & \\
        
          \bottomrule
      \end{tabular}
  \end{center}

  The $\tokone$-leader may change its mind and decide to play
  role~$1$ after noticing the controller does not have its task~$1$ assigned. This switches its role to $-2$. Once the $\tokone$-leader
  notifies the controller, its role is set to $0$ and (in doubt) the
  controller considers that role $2$ is not assigned anymore. The
  $\tokone$-leader then changes its role to $1$. Eventually the
  $\tokone$-leader and $\tokthree$-leader notify the controller that
  roles $1$ and $2$ are taken. This is summarized in
  these three snapshots:
  \begin{center}
    \newcommand{\tabellipsis}{\hspace{-5pt}{\scriptsize $\cdots$}\hspace{-5pt}}
    \hspace{5pt}
    \begin{tabular}{ccccccc}
      \toprule
      
      input & \tabellipsis & 
      $\mathrm{role}$ & \tabellipsis & 
      $\mathrm{task}$ \\
      
      \midrule
      
      $(\tokone, x_1)$ & & & & \\
      $(\tokone, x_1)$ & & & & $[\xmark, \xmark]$ \\
      $(\tokone, x_1)$ & & & & \\
      $(\tokone, x_2)$ & & $-2$ & & \\
      $(\toktwo, x_1)$ & & & & \\
      $(\tokthree, x_2)$ & & $2$ & & \\
      
      \bottomrule
    \end{tabular}\hspace*{20pt}
    \begin{tabular}{cccccc}
      \toprule
      
      \tabellipsis & 
      $\mathrm{role}$ & \tabellipsis & 
      $\mathrm{task}$ \\
      
      \midrule
      
      & & & \\
      & & & $[\xmark, \xmark]$ \\
      & & & \\
      & $0$ & & \\
      & & & \\
      & $2$ & & \\
      
      \bottomrule
    \end{tabular}\hspace*{20pt}
    \begin{tabular}{cccccc}
      \toprule
      
      \tabellipsis & 
      $\mathrm{role}$ & \tabellipsis & 
      $\mathrm{task}$ \\
      
      \midrule
      
      & & & \\
      & & & $[\cmark, \cmark]$ \\
      & & & \\
      & $1$ & & \\
      & & & \\
      & $2$ & & \\
      
      \bottomrule
    \end{tabular}\raisebox{-1.7cm}{\qed}
  \end{center}  
\end{example}
\smallskip

Note that while we rely on stages to describe our protocol, the distributed nature of the model implies that some stages may interfere with others. Therefore, we present here a list of potential problems and the way our protocol fixes them.
\begin{itemize}
  \item While leader election is straightforward, role assignment for leaders can happen at any time before the actual leader is elected. This could lead to the controller being notified of a role assignment for which no current leader is assigned. Thus, when an agent loses its leadership status, it reverts its role to a negative value, meaning it will have to inform the controller of the change before returning to a passive value.
  
  \item A leader may take a role before having the correct counts. We provide a reset mechanism through which the leader falls into a ``negative role''. This forces it to then contact the controller and rectify the situation.
  
  \item A leader may have previously taken a role before realizing it does not actually meet the requirements. The leader is then forced to convey its mistake to the controller. But the controller it notifies may not ultimately be the population's controller. Therefore, after losing the controller status, an agent has to go to a negative controller state, meaning it must reset the controller's tasks before reverting to a passive value.
  
  \item There may be many leaders with the same role. To prevent deadlocks, we allow a leader to self-reassign to a new role if it notices the controller does not have the task filled.
\end{itemize}

\subsubsection{States}

The set of states is defined as
$Q \defeq \{\false, \true\}^{n+2} \times [1..m] \times
[1..r]^{m+1} \times [-n..n] \times \{-1, 0, 1\}$. For the sake of
readability, we specify and manipulate states with these macros:

\begin{center}
  \begin{tabular}{rll}
    \toprule 
    \emph{name} & \emph{values for $q\in Q$} & \emph{values for $q\in I$} \\
    \midrule
    $\init{q}$ & $[1..m]$ & $j \in [1..m]$ \\
    $\val{q}$  & $[1..r]$ & $1$ \\
    $\out{q}$  & $\{\false, \true\}$ & $\false$ \\
    \midrule
    $\lead{q}$ & $\{\false, \true\}$ & $\true$ \\
    $\role{q}$ & $[-n..n]$  & $0$ \\
    $\cnt{\ell}{q}$ & $[1..r]$ &
    $1$ if $\ell = j$, $0$ otherwise \\
    \midrule
    $\ctrl{q}$ & $\{-1, 0, 1\}$ & $1$ \\
    $\task{i}{q}$ & $\{\false, \true\}$ & $\false$ \\
    \bottomrule
  \end{tabular}
\end{center}

The input mapping is defined by $\iota(x_j) \defeq p_j$, where $p_j$
is the unique state of $I$ with $\init{p_j} = j$. Informally,
$\init{q} = j$ indicates that the agent holds the $j$-th element; $\val{q}$ is the
current tally of element $\init{q}$ for the datum of the agent;
$\lead{q}$ and $\ctrl{q}$ respectively indicate whether an agent is a
datum leader or a controller; $\role{q}$ indicates the role for a
leader; $\cnt{j}{q}$ allows a datum leader to maintain the highest
count currently witnessed for element $j$; $\task{i}{q}$ allows the
controller to maintain a list of the currently matched roles; and
$\out{q}$ is the current belief of an agent on the output of the
protocol.

Note that the rules presented in this section are used to succinctly
describe transitions. A single rule may induce several
transitions. Furthermore, for the sake of brevity, we mark rules
allowing \emph{mirror transitions} with an asterisk ($*$) next to the
rule number. Mirror transitions are transitions in which an agent may
observe its own state and react accordingly. Thus, a $*$-rule
generating transitions whose precondition formula is $A(p) \land B(q)$
also generates a transition whose precondition is $A(q) \land B(q)$,
effectively making state $p$ the state of any ``dummy agent''. Note
that $q$ is still the only state to be updated to $q'$.

\subsubsection{Leader and controller election}

The first two rules are meant to elect a unique leader per datum
present in the population, and a unique global controller for the
whole population. For the remainder of the section, let us fix a fair
execution $\c_0 \c_1 \cdots$ where $\c_0$ is initial.

\begin{center}
  \begin{tabular}{cccc}
    \toprule
    \emph{rule}
    & \emph{state precondition}
    & \emph{color precondition}
    & \emph{state update} \\
    
    \midrule\phantomsection\label{int1}
    \multirow{2}{*}{(1)}
    & $\lead{p}$
    & \multirow{2}{*}{$d_1 = d_2$}
    & $\role{q'} = -|\role{q}|$ \\
    & $\lead{q}$
    &
    & $\lnot\lead{q'}$ \\

    \midrule\phantomsection\label{int2}
    \multirow{2}{*}{(2)}
    & $\ctrl{p} = 1$
    & \multirow{2}{*}{\none}
    & \\
    & $\ctrl{q} = 1$
    &
    & $\ctrl{q'} = -1$ \\
    \bottomrule
  \end{tabular}
\end{center}

Note that rule~\intrule{1} guarantees that the agent losing leadership
has its role set to a non-positive value. Similarly, rule~\intrule{2}
pushes the non-controller into a temporary intermediate state for its
controller value, \ie\ $-1$. Let $Q_L \defeq \{q \in Q : \lead{q}\}$
and $Q_C \defeq \{q \in Q :\ctrl{q} = 1\}$. The following lemma
identifies the end of both elections.

\begin{restatable}{lemma}{lemldrctrl}\label{lem:ldr:ctrl}
  There exists $\tau \in \N$ such that
  $\card{Q_C}{\c_{\tau}} = \card{Q_C}{\c_{\tau+1}} = \cdots = 1$, and
  $\card{Q_L}{\c_{\tau}(d)} = \card{Q_L}{\c_{\tau+1}(d)} = \cdots = 1$
  for every $d \in \D$.
\end{restatable}

Let $\alpha$ denote the minimal value $\tau$ given by
\Cref{lem:ldr:ctrl}, which we refer to as the end of the election
stage.

\subsubsection{Element count by datum}

The next rules allow to count how many agents with a common datum hold
the same element. This count is ultimately communicated to the datum
leader. Given $d \in \D$ and $\tau \in \N$, we say that a state $q \in
Q$ is \emph{$(d, j)$-valid} if $\init{q} = j$ and $\val{q} =
\min(\c_\tau(d, Q_j), r)$.

\begin{center}
  \begin{tabular}{cccl}
    \toprule
    \emph{rule}
    & \emph{state precondition}
    & \emph{color precondition}
    & \emph{state update} \\

    \midrule\phantomsection\label{int3}
    \multirow{2}{*}{(3)}
    & $\mathmakebox[35pt][r]{\init{p} =}~\mathmakebox[35pt][l]{\init{q}}$
    & \multirow{2}{*}{$d_1 = d_2$}
    & \\
    & $\mathmakebox[35pt][r]{\val{q} =}~\mathmakebox[35pt][l]{\val{p} < r}$
    &
    & $\val{q'} = \val{q} + 1$ \\

    \midrule\phantomsection\label{int4}
    \multirow{3}{*}{(4)*}
    & $\mathmakebox[60pt][r]{\cnt{\init{p}}{q}}\mathmakebox[30pt][l]{~< \val{p}}$
    & \multirow{3}{*}{$d_1 = d_2$}
    & $\cnt{\init{p}}{q'} = \val{p}$\\
    & $\mathmakebox[60pt][r]{\lead{q}}\mathmakebox[30pt][l]{}$
    &
    & $\texttt{if } (\role{q} > 0 \land {}$ \\

    &
    &
    & \hspace{1.7cm}$\val{p} \notin T(\role{q}, \init{p}))$\texttt{:} \\
    &
    &
    & $\quad \role{q'} = -\role{q}$ \\
    \bottomrule
  \end{tabular}
\end{center}

Observe another correction mechanism; rule~\intrule{4} guarantees that a leader with an assigned role $i > 0$ verifies that it can still assume role $i$ after updating its count. The following lemmas explain that the correct counts are eventually provided to each datum leader.

\begin{restatable}{lemma}{lemcount}\label{lem:count}
  There exists $\tau \in \N$ such that, for
  every $\tau' \geq \tau$, $d \in \supp{\c_{\tau'}}$ and $j \in
  [1..m]$, if $\c_0(d, Q_j) > 0$, then $\c_{\tau'}(d, q) > 0$ holds
  for some $(d, j)$-valid state $q$.
\end{restatable}

\begin{restatable}{lemma}{lemldrcnt}\label{lem:ldrcnt}
  There exists $\tau \geq \alpha$ such
  that, for every $\tau' \geq \tau$, $d \in \supp{\c_{\tau'}}$, $j
  \in [1..m]$ and $q \in \act{\c_{\tau'}(d)}\cap Q_L$, it is the
  case that $\cnt{j}{q} = \min(\c_{\tau'}(d, Q_j), r)$.
\end{restatable}

Let $\tau'$ and $\tau''$ denote the minimal values $\tau$ given by
\Cref{lem:count,lem:ldrcnt}. From now on, let $\beta \defeq
\max(\tau', \tau'')$.

\subsubsection{Role distribution and task tracking}

The following rules assign roles to leaders and allow leaders to reset
their roles when possible, therefore preventing deadlocks. In
rule~\intrule{5}, variable $i$ can take any value from $[1..n]$.
\begin{center}
  \begin{tabular}{cccc}
    \toprule
    \emph{rule}
    & \emph{state precondition}
    & \emph{color precondition}
    & \emph{state update} \\

    \midrule\phantomsection\label{int5}
    \multirow{3}{*}{(5)}
    & $\lead{q}$
    & \multirow{3}{*}{\none}
    & \\
    & $\role{q} = 0$
    &
    & $\role{q'} = i$ \\
    & $\bigwedge_{j \in [1..m]} \cnt{j}{q} \in T(i, j)$
    &
    & \\

    \midrule\phantomsection\label{int6}
    \multirow{5}{*}{(6)*}
    & $\ctrl{p}$
    & \multirow{5}{*}{\none}
    & \\
    & $\lead{q}$
    &
    & \\
    & $\role{q} = i > 0$
    & 
    & $\role{q'} = -i$ \\
    & $\bigvee_{i' \in [1..n] \setminus \{i\}} \big(\lnot\task{i'}{p} \land \hspace{44pt}$
    &
    & \\
    & $\hspace{43pt}\bigwedge_{j \in [1..m]} \cnt{j}{q} \in T(i', j)\big)$
    &
    & \\
    \bottomrule
  \end{tabular}
\end{center}

This induces the following result, informally meaning that if a leader
has taken a role, then it currently \emph{believes} it can fill this
role.

\begin{restatable}{lemma}{lemcorrectrole}\label{lem:correctrole}
  For every $\tau \in \N$, $j \in [1..m]$ and $q \in
  \act{\c_\tau} \cap Q_L$ such that $\role{q} > 0$, it is the
  case that $\cnt{j}{q} \in T(\role{q}, j)$.
\end{restatable}

These rules allow to update the controller's task list and reset roles
when needed:

\begin{center}
  \begin{tabular}{cccc}
    \toprule
    \emph{rule}
    & \emph{state precondition}
    & \emph{color precondition}
    & \emph{state update} \\
    
    \midrule\phantomsection\label{int7}
    \multirow{2}{*}{(7)*}
    & $\role{p} \neq 0$
    & \multirow{2}{*}{\none}
    & $\task{|\role{p}|}{q'} = (\role{p} > 0)$ \\
    & $\ctrl{q} = 1$
    &
    & \\

    \midrule\phantomsection\label{int8}
    \multirow{3}{*}{(8)*}
    & $\ctrl{p} = 1$
    & \multirow{3}{*}{\none}
    & \\
    & $\role{q} < 0$
    &
    & $\role{q'} = 0$ \\
    & $\lnot\task{|\role{q}|}{p}$
    &
    & \\
    \bottomrule
  \end{tabular}
\end{center}
To illustrate how rules~\intrule{5} through~\intrule{8} operate, we give the following example.
\begin{example}
  Recall \Cref{ex:int}, introduced earlier. Consider the configuration of its first snapshot. While we initially gave intuitions on how role reassignment might happen from this specific configuration, we give here a deeper analysis of the important configurations involved in this process.

  \begin{center}
      \begin{tabular}{ccccccc}
        \toprule
        
        input & $\mathrm{val}$ & 
        $\mathrm{lead}$ & $\mathrm{ctrl}$ & 
        $\mathrm{role}$ & $\mathrm{count}$ of $[\# x_1, \# x_2]$ & 
        $\mathrm{task}$ list for [role $1$, role $2$] \\
        
        \midrule
        
        $(\tokone, x_1)$ & $1$ & & & & & \\
        $(\tokone, x_1)$ & $2$ & & \cmark & & & $[\xmark, \cmark]$ \\
        $(\tokone, x_1)$ & $\bf 3$ & & & & & \\
        $(\tokone, x_2)$ & $\bf 1$ & \cmark   & & $2$ & $[3, 1]$ & \\
        $(\toktwo, x_1)$ & $\bf 1$ & \cmark   & &     & $[1, 0]$ & \\
        $(\tokthree, x_2)$ & $\bf 1$ & \cmark & & $2$ & $[0, 1]$ & \\
        
        \bottomrule
      \end{tabular}
  \end{center}

  In the above, the $\tokone$-leader currently believes (rightly so) that it can fill roles $1$ and $2$. Observe that the controller has task $2$ assigned. However, its task $1$ is still unassigned. Therefore, rule~\intrule{6} allows the $\tokone$-leader to initiate its reassignment by setting its role to $-2$ through an interaction with the controller. This leads to the following configuration:
  \begin{center}
    \begin{tabular}{ccccccc}
      \toprule
      
      input & $\mathrm{val}$ & 
      $\mathrm{lead}$ & $\mathrm{ctrl}$ & 
      $\mathrm{role}$ & $\mathrm{count}$ of $[\# x_1, \# x_2]$ & 
      $\mathrm{task}$ list for [role $1$, role $2$] \\
      
      \midrule
      
      $(\tokone, x_1)$ & $1$ & & & & & \\
      $(\tokone, x_1)$ & $2$ & & \cmark & & & $[\xmark, \cmark]$ \\
      $(\tokone, x_1)$ & $\bf 3$ & & & & & \\
      $(\tokone, x_2)$ & $\bf 1$ & \cmark   & & $-2$ & $[3, 1]$ & \\
      $(\toktwo, x_1)$ & $\bf 1$ & \cmark   & &     & $[1, 0]$ & \\
      $(\tokthree, x_2)$ & $\bf 1$ & \cmark & & $2$ & $[0, 1]$ & \\
      
      \bottomrule
    \end{tabular}
\end{center}
Since its role is set to $-2$, the $\tokone$-leader now seeks to inform the controller that it should unassign role $2$ from its task list. This is achieved on their next meeting through rule~\intrule{7}. We then have this next configuration:
\begin{center}
  \begin{tabular}{ccccccc}
    \toprule
    
    input & $\mathrm{val}$ & 
    $\mathrm{lead}$ & $\mathrm{ctrl}$ & 
    $\mathrm{role}$ & $\mathrm{count}$ of $[\# x_1, \# x_2]$ & 
    $\mathrm{task}$ list for [role $1$, role $2$] \\
    
    \midrule
    
    $(\tokone, x_1)$ & $1$ & & & & & \\
    $(\tokone, x_1)$ & $2$ & & \cmark & & & $[\xmark, \xmark]$ \\
    $(\tokone, x_1)$ & $\bf 3$ & & & & & \\
    $(\tokone, x_2)$ & $\bf 1$ & \cmark   & & $-2$ & $[3, 1]$ & \\
    $(\toktwo, x_1)$ & $\bf 1$ & \cmark   & &     & $[1, 0]$ & \\
    $(\tokthree, x_2)$ & $\bf 1$ & \cmark & & $2$ & $[0, 1]$ & \\
    
    \bottomrule
  \end{tabular}
\end{center}

Note that this does not mean that no leader currently has its role set to $2$; indeed, the $\tokthree$-leader still has its role set to $2$. Let us now assume that, immediately after reaching this configuration, the $\tokone$-leader and the controller meet again. Since the $\tokone$-leader observes that the controller no longer has its task $2$ assigned, it can assume that either it unassigned it, some other leader did, or it was never assigned. In any case, it can safely reset its role to $0$ through rule~\intrule{8}, giving us the following configuration:

\begin{center}
  \begin{tabular}{ccccccc}
    \toprule
    
    input & $\mathrm{val}$ & 
    $\mathrm{lead}$ & $\mathrm{ctrl}$ & 
    $\mathrm{role}$ & $\mathrm{count}$ of $[\# x_1, \# x_2]$ & 
    $\mathrm{task}$ list for [role $1$, role $2$] \\
    
    \midrule
    
    $(\tokone, x_1)$ & $1$ & & & & & \\
    $(\tokone, x_1)$ & $2$ & & \cmark & & & $[\xmark, \xmark]$ \\
    $(\tokone, x_1)$ & $\bf 3$ & & & & & \\
    $(\tokone, x_2)$ & $\bf 1$ & \cmark   & & $0$ & $[3, 1]$ & \\
    $(\toktwo, x_1)$ & $\bf 1$ & \cmark   & &     & $[1, 0]$ & \\
    $(\tokthree, x_2)$ & $\bf 1$ & \cmark & & $2$ & $[0, 1]$ & \\
    
    \bottomrule
  \end{tabular}
\end{center}

Observe that the $\tokthree$-leader could have met the controller before the $\tokone$-leader, thereby reassigning role $2$ in the controller's task list and undoing the $\tokone$-leader's work. This would only delay the $\tokone$-leader's role resetting; through fairness, it would not endlessly prevent it.

From this last configuration, since the $\tokone$-leader's role is set to $0$, it is now free to take any role it can fill through rule~\intrule{5}. Let us assume, for the sake of brevity, that it takes on role $1$:

\begin{center}
  \begin{tabular}{ccccccc}
    \toprule
    
    input & $\mathrm{val}$ & 
    $\mathrm{lead}$ & $\mathrm{ctrl}$ & 
    $\mathrm{role}$ & $\mathrm{count}$ of $[\# x_1, \# x_2]$ & 
    $\mathrm{task}$ list for [role $1$, role $2$] \\
    
    \midrule
    
    $(\tokone, x_1)$ & $1$ & & & & & \\
    $(\tokone, x_1)$ & $2$ & & \cmark & & & $[\xmark, \xmark]$ \\
    $(\tokone, x_1)$ & $\bf 3$ & & & & & \\
    $(\tokone, x_2)$ & $\bf 1$ & \cmark   & & $1$ & $[3, 1]$ & \\
    $(\toktwo, x_1)$ & $\bf 1$ & \cmark   & &     & $[1, 0]$ & \\
    $(\tokthree, x_2)$ & $\bf 1$ & \cmark & & $2$ & $[0, 1]$ & \\
    
    \bottomrule
  \end{tabular}
\end{center}
Suppose the $\tokthree$-leader meets the controller before the $\tokone$-leader. Then, rule~\intrule{7} assigns task $2$ in the controller's task list. The $\tokone$-leader can no longer reset its role through rule~\intrule{6} because the controller has task $2$ already assigned. Therefore, when the $\tokone$-leader eventually meets the controller again, it finally assigns task $1$ to its task list via rule~\intrule{7}.
\begin{center}
  \hfill
  \begin{tabular}{ccccccc}
    \toprule
    
    input & $\mathrm{val}$ & 
    $\mathrm{lead}$ & $\mathrm{ctrl}$ & 
    $\mathrm{role}$ & $\mathrm{count}$ of $[\# x_1, \# x_2]$ & 
    $\mathrm{task}$ list for [role $1$, role $2$] \\
    
    \midrule
    
    $(\tokone, x_1)$ & $1$ & & & & & \\
    $(\tokone, x_1)$ & $2$ & & \cmark & & & $[\cmark, \cmark]$ \\
    $(\tokone, x_1)$ & $\bf 3$ & & & & & \\
    $(\tokone, x_2)$ & $\bf 1$ & \cmark   & & $1$ & $[3, 1]$ & \\
    $(\toktwo, x_1)$ & $\bf 1$ & \cmark   & &     & $[1, 0]$ & \\
    $(\tokthree, x_2)$ & $\bf 1$ & \cmark & & $2$ & $[0, 1]$ & \\
    
    \bottomrule
  \end{tabular}\hfill\raisebox{-1.7cm}{\qed}
\end{center}
\end{example}

The following lemmas show that at some point in the execution, a
configuration is reached where agents who are neither leaders nor
controllers no longer interact with other agents.

\begin{restatable}{lemma}{lemnonldrstable}\label{lem:nonldrstable}
  There exists some $\tau\in\N$ such that for every $\tau' \geq \tau$ and $q \in
  \act{\c_{\tau'}} \setminus Q_L$, it is the case that $\role{q}
  = 0$.
\end{restatable}

\begin{center}
  \begin{tabular}{cccc}
    \toprule
    \emph{rule}
    & \emph{state precondition}
    & \emph{color precondition}
    & \emph{state update} \\
    
    \midrule\phantomsection\label{int9}
    \multirow{2}{*}{(9)}
    & \makebox[30pt][l]{$\ctrl{p} = -1$}
    & \multirow{2}{*}{\none}
    & \\
    & \makebox[30pt][l]{$\ctrl{q} = 1$}
    &
    & $\bigwedge_{i \in [1..m]} \lnot\task{i}{q'}$ \\
    
    \midrule\phantomsection\label{int10}
    \multirow{3}{*}{(10)}
    & \makebox[30pt][l]{$\ctrl{p} = 1$}
    & \multirow{3}{*}{\none}
    & \\
    & \makebox[30pt][l]{$\ctrl{q} = -1$}
    &
    & $\ctrl{q'} = 0$ \\
    & $\bigwedge_{i \in[1..m]} \lnot\task{i}{p}$
    &
    & \\
    \bottomrule
  \end{tabular}
\end{center}

\begin{restatable}{lemma}{lemuniquectrl}\label{lem:uniquectrl}
  There exists $\tau \geq \alpha$ such that for every $\tau' \geq \tau$
  and $q \in \act{\c_{\tau'}}$, it is the case that $\ctrl{q} \in
  \{0, 1\}$.
\end{restatable}

Let $\tau'$ and $\tau''$ denote the minimal values $\tau$ given by
\Cref{lem:nonldrstable,lem:uniquectrl}. From now on, let $\gamma \defeq
\max(\beta, \tau', \tau'')$. Informally, this delimits the configuration
where there are no negative controllers, therefore preventing
recurring resets of the controller's tasks through
rule~\intrule{9}. Finally, the following lemma argues that past
$\c_\gamma$, a controller can only have a task set to $\true$ if some
leader is currently assuming the corresponding role (whether positive
or negative).

\begin{restatable}{lemma}{lemctrltask}\label{lem:ctrltask}
  For every $\tau \geq \gamma$, $i \in [1..n]$ and $q \in
  \act{\c_\tau}\cap Q_C$ such that $\task{i}{q}$ holds, there exists
  $q' \in \act{\c_\tau}$ such that $|\role{q'}| = i$.
\end{restatable}

\subsubsection{Output propagation}

The last rule allows the controller to communicate to the other agents
whether it currently has its task list completely assigned or not. Note
that the output of a state $q$ is precisely the value of $\out{q}$, \ie $O(q) \defeq \out{q}$.

\begin{center}
  \begin{tabular}{cccc}
    \toprule
    \emph{rule}
    & \emph{state precondition}
    & \emph{color precondition}
    & \emph{state update} \\
    
    \midrule\phantomsection\label{int11}
    \multirow{1}{*}{(11)*}
    & $\ctrl{p}$
    & \multirow{1}{*}{\none}
    & $\out{q'} = \bigwedge_{i=1}^n \task{i}{p}$ \\
    \bottomrule
  \end{tabular}
\end{center}

\begin{restatable}{lemma}{lemint}\label{lem:int}
  It is the case that $\psi(\c_0)$ holds iff there exists some
  $\tau \geq \gamma$ such that for every $\tau' \geq \tau$, there
  exists $q \in \act{\c_{\tau'}} \cap Q_C$ such that $\bigwedge_{i \in
  [1..n]} \task{i}{q}$ holds.
\end{restatable}

\begin{restatable}{corollary}{thminterval}\label{thm:interval}
  There exists $\tau \in \N$ such that $O(\c_\tau) = O(\c_{\tau+1}) =
  \cdots = \psi(\c_0)$.
\end{restatable}

\section{Conclusion}\label{sec:conclusion}
In this article, we introduced population protocols with unordered
data; we presented such a protocol that computes majority over an
infinite data domain; and we established the expressive power of
immediate observation protocols: they compute interval predicates.

This work initiates the study of population protocols operating over
arbitrarily large domains. Hence, this opens the door to numerous
exciting questions, \eg\ on space-efficient and time-efficient
protocols. In particular, the expressive power of our model remains
open.

There exist results on logics over data multisets (\eg,
see~\cite{Lug05,PK08}). In particular, the author of~\cite{Lug05}
provides a decidable logic reminiscent of Presburger arithmetic. It
appears plausible that population protocols with unordered data
compute (perhaps precisely) this logic. While we are fairly confident
that remainder and threshold predicates with respect to the data counts can be computed in our model,
the existential quantification, arising in the (non-ambiguous) solved
forms of~\cite{Lug05}, seems more challenging to implement than the
one of simple interval predicates.

Our model further relates to logic and automata on data words: inputs
of a protocol with data can be seen as data words where $Q$ is the
alphabet and $\D$ is the data domain. Importantly, these data words
are commutative, \ie, permutations do not change acceptance. For
example, the logic $\mathsf{FO}^2(+1, \sim, <)$ of~\cite{BDMSS11}
allows to specify non-commutative properties such as ``there is a
block of $a$'s followed by a block of $b$'s''. In this respect, this
logic is too ``strong''. It is also too ``weak'' as it cannot express
``for each datum, the number of $a$'s is even''. For this same reason,
$\mathsf{EMSO}^2(+1,
\sim)$, and equivalently weak data automata~\cite{KST12}, is too
``weak''. The logic $\mathsf{EMSO}^2_\#(+1, \sim)$, and equivalently
commutative data automata~\cite{Wu12}, can express the latter, but,
again, the successor relation allows to express non-commutative
properties on letters. Thus, while models related to data words have
been studied and could influence research on the complete
characterization of the expressive power of our model, we have yet to
directly connect them to our model.

\bibliographystyle{plain}
\bibliography{references}

\clearpage
\appendix
\label{appendix}

\section{Missing proofs of Section~\ref{sec:maj}}
Let us first recall all of the transitions in a single table:

\begin{center}
  \begin{tabular}{cccc}
    \bottomrule
    \multicolumn{4}{c}{\textbf{Pairing stage}} \\
    \toprule
    \emph{rule}
    & \emph{state precondition}
    & \emph{color precondition}
    & \emph{state update} \\
    \midrule

    \multirow{2}{*}{(1)}
    & $\lnot\pair{p}$
    & \multirow{2}{*}{$d_1 \neq d_2$}
    & $\pair{p'} \land \even{p'}$ \\

    & $\lnot\pair{q}$
    &
    & $\pair{q'} \land \even{q'}$ \\

    \bottomrule

    \multicolumn{4}{c}{\textbf{Grouping stage}} \\
    \toprule
    \emph{rule}
    & \emph{state precondition}
    & \emph{color precondition}
    & \emph{state update} \\
    \midrule

    \multirow{2}{*}{(2)}
    & \makebox[130pt][l]{$\lnot\pair{p}$}
    & \multirow{2}{*}{$d_1 \neq d_2$}
    & \none \\
    & \makebox[130pt][l]{$\phantom{\lnot}\pair{q} \land \phantom{\lnot}\grp{q} \land \maj{q} = Y$}
    &
    & \makebox[100pt][r]{$\lnot\grp{q'} \land \maj{q'} = N$} \\

    \midrule      
    \multirow{2}{*}{(3)}
    & \makebox[130pt][l]{$\lnot\pair{p}$}
    & \multirow{2}{*}{$d_1 = d_2$}
    & \none \\
    & \makebox[130pt][l]{$\phantom{\lnot}\pair{q} \land \lnot\grp{q} \land \maj{q} = N$}
    &
    & \makebox[100pt][r]{$\phantom{\lnot}\grp{q'} \land \maj{q'} = Y$} \\

    \midrule    
    \multirow{2}{*}{(4)}
    & \makebox[130pt][l]{$\lnot\pair{p}$}
    & \multirow{2}{*}{$d_1 \neq d_2$}
    & \none \\
    & \makebox[130pt][l]{$\phantom{\lnot}\pair{q} \land \phantom{\lnot}\grp{q} \land \maj{q} \in \{y, n\}$}
    &
    & \makebox[100pt][r]{$\maj{q'} = \overline{N}$} \\
    \midrule
    
    \multirow{2}{*}{(5)}
    & \makebox[130pt][l]{$\lnot\pair{p}$}
    & \multirow{2}{*}{$d_1 = d_2$}
    & \none \\
    & \makebox[130pt][l]{$\phantom{\lnot}\pair{q} \land \lnot\grp{q} \land \maj{q} \in \{y, n\}$}
    &
    & \makebox[100pt][r]{$\maj{q'} = \overline{Y}$} \\

    \midrule \\[-16.5pt]
    \midrule
    \multirow{2}{*}{(6)}
    & \makebox[130pt][l]{$\phantom{\lnot\pair{p} \land \lnot}\grp{p} \land \maj{p} = \overline{N}$}
    & \multirow{2}{*}{\none}
    & \makebox[100pt][r]{$\lnot\grp{p'} \land \maj{p'} = N$} \\
    & \makebox[130pt][l]{$\phantom{\lnot\pair{q} \land {}} \lnot\grp{q} \land \maj{q} \in \{y, n\}$}
    &
    & \makebox[100pt][r]{$\maj{q'} = N$} \\

    \midrule
    \multirow{2}{*}{(7)}
    & \makebox[130pt][l]{$\phantom{\lnot\pair{p} \land {}} \lnot\grp{p} \land \maj{p} = \overline{Y}$}
    & \multirow{2}{*}{\none}
    & \makebox[100pt][r]{$\phantom{\lnot}\grp{p'} \land \maj{p'} = Y$} \\
    & \makebox[130pt][l]{$\phantom{\lnot\pair{q} \land {}}\phantom{\lnot}\grp{q} \land \maj{q} \in \{y, n\}$}
    &
    & \makebox[100pt][r]{$\maj{q'} = Y$} \\

    \midrule
    \multirow{2}{*}{(8)}
    & \makebox[130pt][l]{$\phantom{\lnot\pair{p} \land {}}\phantom{\lnot}\grp{p} \land \maj{p} = \overline{N}$}
    & \multirow{2}{*}{\none}
    & \makebox[100pt][r]{$\lnot\grp{p'} \land \maj{p'} = n$} \\
    & \makebox[130pt][l]{$\phantom{\lnot\pair{q} \land {}}\lnot\grp{q}\land \maj{q} = \overline{Y}$}
    &
    & \makebox[100pt][r]{$\phantom{\lnot}\grp{q'} \land \maj{q'} = n$} \\
    \bottomrule

    \multicolumn{4}{c}{\textbf{Majority stage}} \\
    \toprule
    \emph{rule}
    & \emph{state precondition}
    & \emph{color precondition}
    & \emph{state update} \\
    \midrule

    \multirow{2}{*}{(9)}
    & $\lnot\pair{p}\phantom{\land \lnot\even{q}}$
    & \multirow{2}{*}{\none}
    & \none \\
    & $\phantom{\lnot\pair{p} \land \lnot}\even{q}$
    &
    & $\lnot\even{q'}$ \\

    \midrule
    \multirow{2}{*}{(10)}
    & $\phantom{\lnot}\pair{p} \land \phantom{\lnot}\even{p}$
    & \multirow{2}{*}{\none}
    & \none \\
    & $\phantom{\lnot}\pair{q} \land \lnot\even{q}$
    &
    & $\phantom{\lnot}\even{q'}$ \\

    \midrule \\[-16.5pt]
    \midrule
    \multirow{2}{*}{(11)}
    & \makebox[50pt][l]{$\maj{p} = Y$}
    & \multirow{2}{*}{\none}
    & \makebox[50pt][l]{$\maj{p'} = n$} \\
    & \makebox[50pt][l]{$\maj{q} = N$}
    &
    & \makebox[50pt][l]{$\maj{q'} = n$} \\

    \midrule
    \multirow{2}{*}{(12)}
    & \makebox[50pt][l]{$\maj{p} = Y$}
    & \multirow{2}{*}{\none}
    & \none \\
    & \makebox[50pt][l]{$\maj{q} = n$}
    &
    & \makebox[50pt][l]{$\maj{q'} = y$} \\

    \midrule
    \multirow{2}{*}{(13)}
    & \makebox[50pt][l]{$\maj{p} = N$}
    & \multirow{2}{*}{\none}
    & \none \\
    & \makebox[50pt][l]{$\maj{q} = y$}
    &
    & \makebox[50pt][l]{$\maj{q'} = n$} \\

    \midrule
    \multirow{2}{*}{(14)}
    & \makebox[50pt][l]{$\maj{p} = n$}
    & \multirow{2}{*}{\none}
    & \none \\
    & \makebox[50pt][l]{$\maj{q} = y$}
    &
    & \makebox[50pt][l]{$\maj{q'} = n$} \\
    \bottomrule
  \end{tabular}
\end{center}
Note that for the sake of brevity, for the remainder of the section, we use the notation $\c\reach{(n)}\c'$ to denote that $\c\reach{t}\c'$ using a transition $t$ arising from rule $(n)$.
\lempair*

\begin{proof}\label{prf:pair}
  First, note that no rule reverses a state from paired to
  unpaired. This implies that for any two $\c, \c'$ such that $\c
  \reach{*} \c'$, it is the case that $\card{U}{\c} \geq
  \card{U}{\c'}$. This yields a threshold $\tau \in \N$ such that
  $\card{U}{\c_\tau} = \card{U}{\c_{\tau+1}} = \cdots$, as the value
  cannot drop below zero.

  Assume that there exist infinitely many $i \geq \tau$ such
  that $\c_i(d, U) > 0$ and $\c_i(d', U) > 0$ holds for two distinct colors $d, d' \in \D$. Rule~\majrule{1} is enabled in each such $\c_i$. This implies
  the existence of a configuration $\c_i'$ such that $\c_i \reach{*}
  \c_i'$ and $\card{U}{\c_i} > \card{U}{\c_i'}$. By fairness, some
  configuration from $\{\c_i' : i \in \N\}$ is reached along the
  execution, which is a contradiction.
\end{proof}

\lemmajcolor*

\begin{proof}\label{prf:majcolor}
  Since $d$ is the majority color and all agents are initially
  unpaired, we have
  \begin{align}
    \c_0(d, U) > \sum_{d' \neq d} \c_0(d', U).\label{eq:maj:card}
  \end{align}
  Moreover, for every $i > 0$ such that $\c_{i-1}\reach{(1)} \c_i$, at
  most of one state $q \in \act{\c_{i-1}(d)}$ is paired while at least
  one color $d'\neq d$ has one of its states $q' \in
  \act{\c_{i-1}(d')}$ paired. Therefore, we either have $\c_i(d, U) =
  \c_{i-1}(d, U)$ or $\c_i(d, U) = \c_{i-1}(d, U) - 1$, and
  we have
  \[
  \sum_{d'\neq d} \c_i(d', U)
  <
  \sum_{d'\neq d} \c_{i-1}(d', U).
  \]
  Consequently, by~\eqref{eq:maj:card}, we have $\c_i(d, U) > \sum_{d'
    \neq d} \c_i(d', U) \geq 0$.

  Recall that only rule~\majrule{1} changes the pairing of a
  state. Thus, $\c_i(d, U) > 0$ holds for all $i
  \in \N$.
\end{proof}

\leminvariant*

\begin{proof}\label{prf:invariant}
  We prove the claim by induction on $i \geq 0$. We have
  $\card{Q_N}{\c_0} = \card{Q_m}{\c_0} = 0$ and $\card{Q_Y}{\c_0} =
  \card{Q_M}{\c_0} = |\c_0|$. Therefore, the claim holds initially.
  
  Given a transition $t = ((p, q), \sim, (p', q')) \in \delta$ and $x
  \in \{Y, N, M, m\}$, we define $\Delta_x(t) \defeq \multiset{p', q'}(Q_x)
  - \multiset{p, q}(Q_x)$. This table illustrates the value
  of $\Delta_x(t)$ for each $t$:
  \begin{center}\label{tbl:invariant}
    \begin{tabular}{ccccccc}
      \toprule
      $t\in\delta$ from rule \# & $\Delta_Y(t)$ & $\Delta_N(t)$ & $\Delta_M(t)$ & $\Delta_m(t)$ & $\Delta_Y(t) - \Delta_N(t)$ & $\Delta_M(t) - \Delta_m(t)$ \\
      \midrule
      \majrule{2} & $-1$ & $1$ & $-1$ & $1$ & $-2$ & $-2$ \\
      \midrule
      \majrule{3} & $1$ & $-1$ & $1$ & $-1$ & $2$ & $2$ \\
      \midrule
      \majrule{6} & $0$ & $2$ & $-1$ & $1$ & $-2$ & $-2$ \\
      \midrule
      \majrule{7} & $2$ & $0$ & $1$ & $-1$ & $2$ & $2$ \\
      \midrule
      \majrule{11} & $-1$ & $-1$ & $0$ & $0$ & $0$ & $0$ \\
      \midrule
      other rules & $0$ & $0$ & $0$ & $0$ & $0$ & $0$ \\
      \bottomrule
    \end{tabular}
  \end{center}
  Let $i > 0$. Observe that, by definition, $\card{Q_x}{\c_i} =
  \card{Q_x}{\c_{i-1}} + \Delta_x(t)$ holds for every $t \in \delta$
  and $x \in \{Y, N, M, n\}$. Thus,
  \begin{align*}
    \card{Q_Y}{\c_i} - \card{Q_N}{\c_i} 
    &= (\card{Q_Y}{\c_{i-1}} + \Delta_Y(t)) - (\card{Q_N}{\c_{i-1}} + \Delta_N(t)) \\
    &= \big(\card{Q_Y}{\c_{i-1}} - \card{Q_N}{\c_{i-1}}\big) + \big(\Delta_Y(t) - \Delta_N(t)\big) \\
    &= \big(\card{Q_M}{\c_{i-1}} - \card{Q_m}{\c_{i-1}}\big) + \big(\Delta_Y(t) - \Delta_N(t)\big)
    && \text{(by induction hyp.)} \\
    &= \big(\card{Q_M}{\c_{i-1}} - \card{Q_m}{\c_{i-1}}\big) + \big(\Delta_M(t) - \Delta_m(t)\big)
    && \text{(by above table)} \\
    &= \big(\card{Q_M}{\c_{i-1}} + \Delta_M(t)\big) - \big(\card{Q_m}{\c_{i-1}} + \Delta_m(t) \big) \\
    &= \card{Q_M}{\c_i} - \card{Q_m}{\c_i}. &&\qedhere
  \end{align*}
\end{proof}

\lempartner*

\begin{proof}\label{prf:partner}
  We proceed by induction on $i \in \N$. The claim trivially holds for
  $\c_0$, as the unique initial state $q_I$ satisfies $\maj{q_I}
  \notin E$, which implies $\card{E}{\c_0} = \card{E_M}{\c_0} =
  \card{E_m}{\c_0} = 0$.

  Let $i > 0$ and let $t$ be such that $\c_{i-1} \reach{t} \c_i$. For
  $\card{E_M}{\c_i}$ (resp.\ $\card{E_m}{\c_i}$) to be different from
  $\card{E_M}{\c_{i-1}}$ (resp.\ $\card{E_m}{\c_{i-1}}$), $t$ must
  be~\majrule{6},~\majrule{7} or~\majrule{11}. Let us take a look at
  the effect of each of these transitions:
  \begin{center}\label{tbl:partner}
    \begin{tabular}{ccc}
      \toprule
        $t\in\delta$ from rule \# & $\card{E_M}{\c_i} - \card{E_M}{\c_{i-1}}$ & $\card{E_m}{\c_i} -\card{E_m}{\c_{i-1}}$ \\
        \midrule
        \majrule{6} & $-1$ & $-1$ \\
        \midrule
        \majrule{7} & $-1$ & $-1$ \\
        \midrule
        \majrule{11} & $1$ & $1$ \\
      \bottomrule
    \end{tabular}
  \end{center}
  Observe that rule~\majrule{11} does not explicitly specify the groups of states $p$ and $q$. However, no rule generates a state $q$ such that $\maj{q} = Y$ (resp.\ $\maj{q} = N$) and $\lnot\grp{q}$ (resp.\ $\grp{q}$), and $q_I$ satisfies $\maj{q_I} = Y \land \grp{q_I}$; therefore, for states $p$ and $q$ of rule~\majrule{11}, $\grp{p}$ and $\lnot\grp{q}$ implicitly hold. Since the effect is the same on each side of the equation for all
  transitions, the claim follows by induction hypothesis.
\end{proof}

\lemgroup*

\begin{proof}\label{prf:group}
  We claim that there exist $\tau, \tau', \tau'', \tau''' \geq \alpha$
  such that
  \begin{multicols}{2}
    \begin{enumerate}
    \item $\c_i(d, Q_{\overline{N}}) = 0$ for all $i \geq \tau$;
      
    \item $\sum_{d' \neq d} \c_i(d', Q_{\overline{Y}}) = 0$ for all $i
      \geq \tau'$;
      
    \item $\c_i(d, Q_m) = 0$ for all $i \geq \tau''$;
      
    \item $\sum_{d' \neq d} \c_i(d', Q_M) = 0$ for all $i \geq \tau'''$.
    \end{enumerate}
  \end{multicols}
  \noindent Before proving the claims, observe that the lemma follows
  by taking $\max(\tau, \tau', \tau'', \tau''').$

  \medskip
  \noindent\emph{Claim~1.} By definition of $\alpha$ and
  \Cref{lem:pair}, there are no unpaired agents of a color $d' \neq d$ in $\c_\alpha$
  to play the corresponding role in rule~\majrule{4}. Moreover,
  rule~\majrule{4} is the only rule generating states from
  $Q_{\overline{N}}$. Thus, for every $i \geq \alpha$ and $\c_i
  \reach{*} \c$, we have $\c_i(d, Q_{\overline{N}}) \geq \c(d,
  Q_{\overline{N}})$.

  It remains to show that the value decreases to zero. If $\c_i(d,
  Q_{\overline{N}}) > 0$ holds for only finitely many indices $i \geq
  \alpha$, then we are done. Otherwise, we must have $\c_i(d, Q_{\overline{N}} \cap Q_M) > 0$ for infinitely many
  $i \geq \alpha$. Indeed, rule~\majrule{4} updates $\maj{q'} =
  \overline{N}$ under the precondition $\grp{q}$, and no rule produces
  a state from $Q_{\overline{N}} \cap Q_m$. By \Cref{lem:partner},
  this implies the existence of a state from $\act{\c_i} \cap E_m$
  enabling rule~\majrule{6} or rule~\majrule{8}. Thus, there exists a configuration $\c_i'$ such
  that $\c_i \reach{*} \c_i'$ and $\c_i(d, Q_{\overline{N}}) >
  \c_i'(d, Q_{\overline{N}})$. Hence, by fairness, we obtain
  $\c_\tau(d, Q_{\overline{N}}) = 0$ for some $\tau \geq
  \alpha$. Hence, $\c_\tau(d, Q_{\overline{N}}) = \c_{\tau+1}(d,
  Q_{\overline{N}}) = \cdots = 0$ (as the value cannot increase).

  \medskip
  \noindent\emph{Claim~2.} Similarly to Claim~1, as unpaired agents
  share color $d$, rule~\majrule{5} is only enabled for agents of
  color $d$ for $\c_\alpha$. So, as rule~\majrule{5} is the only
  rule generating states from $Q_{\overline{Y}}$, we have
  \[
  \sum_{d' \neq d} \c_i(d', Q_{\overline{Y}})
  \geq
  \sum_{d' \neq d} \c(d', Q_{\overline{Y}})
  \text{ for every }
  i \geq \alpha \text{ and } \c_i \reach{*} \c.
  \]
  It remains to show that the value decreases to zero. If $\sum_{d'
    \neq d} \c_i(d', Q_{\overline{Y}}) > 0$ for only finitely many indices $i
  \geq \alpha$, then we are done. Otherwise, we must have
  $\c_i(d, Q_{\overline{Y}} \cap Q_m) > 0$ for
  infinitely many $i \geq \alpha$. Indeed, rule~\majrule{5} updates
  $\maj{q'} = \overline{Y}$ under the precondition $\neg \grp{q}$, and
  no rule produces a state from $Q_{\overline{Y}} \cap Q_M$. By
  \Cref{lem:partner}, some
  state from $\act{\c_i} \cap E_M$ enables rule~\majrule{7} or
  rule~\majrule{8}. Thus, there exists a configuration $\c_i'$ such
  that $\c_i \reach{*} \c_i'$ and
  \[
  \sum_{d' \neq d} \c_i(d', Q_{\overline{Y}})
  >
  \sum_{d' \neq d} \c_i'(d', Q_{\overline{Y}}).
  \]
  Thus, by fairness, we obtain $\sum_{d' \neq d} \c_{\tau'}(d',
  Q_{\overline{Y}}) = 0$ for some $\tau' \geq \alpha$. So,
  $\sum_{d' \neq d} \c_{\tau'}(d', Q_{\overline{Y}}) = \sum_{d' \neq
    d} \c_{\tau'+1}(d', Q_{\overline{Y}}) = \cdots = 0$ (as the value
  cannot increase).

  \medskip
  \noindent\emph{Claim~3.} Let $\alpha' \defeq \max(\tau, \tau')$. By
  Claim~1, from $\c_{\alpha'}$, since no agent of color $d$ is in a
  state from $Q_{\overline{N}}$, rules~\majrule{6} and~\majrule{8} can
  no longer change the group of an agent of color $d$ to the minority
  group, \ie\ set $\grp{q}$ to $\false$. Moreover, since agents of
  color $d' \neq d$ are all paired by \Cref{lem:pair},
  rule~\majrule{2} can also no longer change an agent of color $d$ to
  the minority group. Therefore, for any $\c'$ such that $\c_{\alpha'}
  \reach{*} \c'$, we have $\c_{\alpha'}(d, Q_m) \geq \c'(d, Q_m)$.

  If $\c_i(d, Q_m) > 0$ for only finitely many indices $i \geq
  \alpha'$, then we are done. Thus, assume that $\c_i(d, Q_m) > 0$ for
  infinitely many $i \geq \alpha'$. If there exists $q \in
  \act{\c_i(d)} \cap Q_m \cap P$, then, depending on the value of
  $\maj{q}$, either rule~\majrule{3}; rule~\majrule{5} followed by rule~\majrule{7};
  or rule~\majrule{5} followed by rule~\majrule{8} is enabled. In each case, it
  yields $\c_i'$ such that $\c_i \reach{*} \c_i'$ and $\c_i(d, Q_m
  \cap P) > \c_i'(d, Q_m \cap P)$. Thus, by fairness, we obtain
  $\tau'' \geq \alpha$ such that $\c_i(d, Q_m \cap P) = 0$ for all $i
  \geq \tau''$.

  It remains to show that $\c_{\tau''}(d, Q_m \cap U) =
  0$. Rules~\majrule{2} and~\majrule{3} have the precondition
  $\pair{q}$ for the only modified state $q'$. Hence, they cannot
  change the group of an unpaired agent. While
  rules~\majrule{6},~\majrule{7} and~\majrule{8} do not have such a
  precondition, they require state $p$ to satisfy $\maj{p} \in
  \{\overline{Y}, \overline{N}\}$. The only rules changing a state to
  satisfy that condition are~\majrule{4} and~\majrule{5}. Since both
  rules only change a state $q$ to $q'$ satisfying $\maj{q'}\in\{
  \overline{Y}, \overline{N}\}$ if $\pair{q}$ holds, unpaired agents
  are never in a state meeting the requirements to have their group
  changed. From these observations, we conclude that a state $q \in
  \act{\c_{\tau''}}$ cannot satisfy both $\lnot\grp{q}$ and
  $\lnot\pair{q}$. Thus, we have $\c_{\tau''}(d, Q_m \cap U) = 0$ as
  desired.

  \medskip
  \noindent\emph{Claim~4}. Let $\alpha' \defeq \max(\tau, \tau')$. By
  Claim~2, from $\c_{\alpha'}$, no agent of color $d'\neq d$ is in a
  state from $Q_{\overline{Y}}$. Therefore, rules~\majrule{7}
  and~\majrule{8} can no longer change the group of an agent of color
  $d' \neq d$ to the majority group. Since there is no unpaired agent
  of color $d' \neq d$ by \Cref{lem:pair}, rule~\majrule{3} can also
  no longer change their group to the majority group. So, for any
  $\c'$ such that $\c_{\alpha'} \reach{*} \c'$ and for any $d' \neq
  d$, we have $\c_{\alpha'}(d', Q_M) \geq \c'(d', Q_M)$.

  If $\sum_{d' \neq d} \c_i(d', Q_M) > 0$ for only finitely many
  indices $i \geq \alpha'$, then we are done. Thus, assume that there
  exist infinitely many $i \geq \alpha'$ such that for $d' \neq d$, there exists $q
  \in \act{\c_{i}(d')} \cap Q_M$. Depending on the value of
  $\maj{q}$, either rule~\majrule{2}; rule~\majrule{4} followed by rule~\majrule{6}; or rule~\majrule{4} followed by rule~\majrule{8} is enabled. In each
  case, there exists a configuration $\c_i'$ such that $\c_i \reach{*}
  \c_i'$ and $\c_i(d', Q_M) > \c_i'(d', Q_M)$. Therefore, we are done
  by fairness.
\end{proof}

\lemeven*

\begin{proof}\label{prf:even}
  Note that from $\c_\alpha$ onwards, only rules~\majrule{9}
  and~\majrule{10} may change a state $q$ to a state $q'$ such that
  $\even{q'} \neq \even{q}$. Let $Q_e \defeq \{q \in Q :
  \even{q}\}$. We make a case distinction.

  \medskip
  \noindent\emph{Case~$\card{U}{\c_\alpha} > 0$}. Suppose that there exist infinitely many $i \geq
  \alpha$ such that, for some $q\in \act{\c_i}$, $\even{q}$ holds, as
  we are otherwise done. As rule~\majrule{9} is enabled, there exists
  $\c_i'$ such that $\c_i \reach{*} \c_i'$ and $\card{Q_e}{\c_i} >
  \card{Q_e}{\c_i'}$. By repeatedly invoking fairness, we obtain $\tau
  \geq \alpha$ such that $\c_\alpha \reach{*} \c_\tau$ and
  $\card{Q_e}{\c_\tau} = 0$. Moreover, for any $\c'$ such that
  $\c_\tau \reach{*} \c'$, rules \majrule{9} and \majrule{10} are
  permanently disabled since $\card{Q_e}{\c'} = 0$. So, for any $q \in \act{\c'}$, we have
  $\lnot\even{q}$.

  \medskip
  \noindent\emph{Case~$\card{U}{\c_\alpha} = 0$}. Since
  every agent is initially unpaired, we necessarily have $\alpha >
  0$. For every $i \geq \alpha$, note that rule~\majrule{9} is
  disabled in $\c_i$. Therefore, only rule~\majrule{10} can change a
  state $q$ to a state $q'$ such that $\even{q'} \neq \even{q}$. By
  minimality of $\alpha$, we have
  \[
  \c_{\alpha-1} \reach{(1)} \c_\alpha,
  \text{ and hence }
  \card{Q_e}{\c_\alpha} \geq 2.
  \]
  Since rule~\majrule{9} is permanently disabled,
  we have $\card{Q_e}{\c_i} > 0$ for every $i \geq \alpha$.

  If $\act{\c_i} \not\subseteq Q_e$ for only finitely many indices $i
  \geq \alpha$, then we are done. So, assume that there exist
  infinitely many $i$ such that for some $q \in \act{\c_i}$,
  $\lnot\even{q}$ holds. Rule~\majrule{10} is enabled, which means
  that there exists $\c_i'$ such that $\c_i \reach{*} \c_i'$ and
  $\card{Q_e}{\c_i} < \card{Q_e}{\c_i'}$. Therefore, by repeatedly
  invoking fairness, we obtain some $\tau \geq \alpha$ such that
  $\c_\alpha \reach{*} \c_\tau$ and $\card{Q_e}{\c_\tau} = |\c_\tau|$. Thus, $\act{\c_i} \subseteq Q_e$ holds
  for all $i \geq \tau$.
\end{proof}

\lemconvtrue*

\begin{proof}\label{prf:convtrue}
  Let $d \in \D$ be the majority color. By \Cref{lem:majcolor}, we
  have $\c_\alpha(d, U) > 0$. As there is a
  majority for $d$, we have $\c_i(d, Q) > \sum_{d' \neq d} \c_i(d',
  Q)$ for all $i \in \N$. Let $\tau$ be given by \Cref{lem:group} for
  color $d$. The following holds for all $i \geq \tau$:
  \[
  \act{\c_i(d)} \subseteq Q_M
  \text{ and }
  \act{\c_i(d')} \subseteq Q_m \text{ for all } d' \neq d.
  \]
  We can infer from the previous statements that $\card{Q_M}{\c_i} >
  \card{Q_m}{\c_i}$ for all $i \geq \tau$. Combined with
  \Cref{lem:invariant}, we conclude that $\card{Q_Y}{\c_i} >
  \card{Q_N}{\c_i}$ for all $i \geq \tau$.

  Let us show that the number of active states from $Q_N$ decreases
  permanently to zero. Assume that there are infinitely many indices
  $i \geq \tau$ such that $\card{Q_N}{\c_i} > 0$, as we
  are otherwise done. As there is at least one state $q\in\act{\c_i}\cap Q_N$ and $\card{Q_Y}{\c_i} > \card{Q_N}{\c_i}$, rule~\majrule{11} is enabled. Thus, there exists
  $\c_i'$ such that $\c_i \reach{*} \c_i'$ and $\card{Q_N}{\c_i} >
  \card{Q_N}{\c_i'}$. So, by repeatedly invoking fairness, we obtain
  some $\tau' \geq \tau$ such that $\card{Q_N}{\c_{\tau'}} = 0$. From
  $\c_{\tau'}$ onwards, rules~\majrule{11} and~\majrule{13} are
  disabled.

  Let us now show that the number of active states from $Q_n$
  decreases permanently to zero. Assume that there exist infinitely
  many $i \geq \tau'$ such that $\card{Q_n}{\c_i} > 0$,
  as we are otherwise done. Since there is at least one state $q\in\act{\c_i}\cap Q_Y$, rule~\majrule{12} is enabled. Therefore, there
  exists $\c_i'$ such that $\c_i \reach{*} \c_i'$ and
  $\card{Q_n}{\c_i} > \card{Q_n}{\c_i'}$. By repeatedly invoking
  fairness, we obtain some $\tau''$ such that $\card{Q_n}{\c_{\tau''}}
  = 0$. From there, rules~\majrule{12} and~\majrule{14} are
  disabled.
  
  Let us now prove that there are no active states $q$ in $\c_\tau$ (and any $\c'$ reachable from $\c_\tau$) such that $\maj{q} \in \{\overline{Y}, \overline{N}\}$. Recall the claims from~\Cref{lem:group}; agents of $d$ are no longer in states from $Q_{\overline{N}}$ and agents of $d'\neq d$ are no longer in states of $Q_{\overline{Y}}$. Furthermore, all agents are in the correct group, \ie agents of $d$ are in the majority group while agents of $d'\neq d$ are in the minority group. Assume that for $i\geq\tau$, an agent of $d'\neq d$ (resp.\ $d$) is in a state $q\in\act{\c_i}$ such that $\maj{q} = \overline{N}$ (resp.\ $\maj{q} = \overline{Y}$). This implies the existence of a configuration $\c'$ reachable from $\c_i$ such that some state $q'\in\act{\c'(d')}$ (resp.\ $q'\in\act{\c'(d)}$) satisfies $\grp{q'}$ (resp.\ $\lnot\grp{q'}$). By fairness, such a configuration is reached, which is a contradiction to the definition of $\tau$ from~\Cref{lem:group}. Thus, for any state $q \in \act{\c_{\tau''}}$, we have
  $\maj{q} \in \{Y, y\}$.

  Finally, by \Cref{lem:even}, there exists $\tau''' \geq \tau''$ such
  that for every $i \geq \tau'''$ and $q \in \act{\c_i}$, it is the
  case that $\lnot\even{q}$ holds.
\end{proof}

\lemconvfalse*

\begin{proof}\label{prf:convfalse}
  If $\card{U}{\c_\alpha} = 0$, then, by
  \Cref{lem:pair,lem:even}, there exists $u \geq \alpha$ such that
  for any $i \geq u$, each state $q \in \act{\c_i}$ satisfies
  $\even{q}$.

  Now, assume this is not the case. By \Cref{lem:pair}, there is a
  unique color $d \in \D$ such that $\card{U}{\c_\alpha(d)} > 0$. As there is no majority color, we have $\c_i(d, Q) \leq
  \sum_{d' \neq d} \c_i(d', Q)$ for every $i \in \N$. Let $\tau \geq \alpha$ be given by \Cref{lem:group} and taken as minimal. By \Cref{lem:group}, for every $i \geq \tau$:
  \[
  \act{\c_i(d)} \subseteq Q_M
  \text{ and }
  \act{\c_i(d')} \subseteq Q_m
  \text{ for all } d' \neq d.
  \]
  So, we have $\card{Q_M}{\c_i} \leq \card{Q_m}{\c_i}$ for all $i \geq
  \tau$. By \Cref{lem:invariant}, this implies that $\card{Q_Y}{\c_i}
  \leq \card{Q_N}{\c_i}$ for all $i \geq \tau$.

  Let us show that the number of active states from $Q_Y$ decreases
  permanently to zero. Assume that there exist infinitely many indices
  $i \geq \tau$ such that $\card{Q_Y}{\c_i} > 0$. As there is at least one state $q\in\act{\c_i}\cap Q_Y$ and $\card{Q_Y}{\c_i} \leq \card{Q_N}{\c_i}$, rule~\majrule{11} is enabled. So, there exists $\c_i'$ such that $\c_i
  \reach{*} \c_i'$ and $\card{Q_Y}{\c_i} >
  \card{Q_Y}{\c_i'}$. By repeatedly invoking fairness, we
  obtain some $\tau' \geq \tau$ with $\card{Q_Y}{\c_{\tau'}} =
  0$. Let $\tau'$ be minimal. Rules~\majrule{11} and~\majrule{12} are now
  disabled.

  Since rule~\majrule{12} is disabled, we have $\card{Q_y}{\c_i} \geq
  \card{Q_y}{\c_{i+1}}$ for all $i \geq \tau'$. Let us show that the
  number of active states from $Q_y$ decreases permanently to
  zero. Two subcases arise.

  First, consider the case where $\card{Q_N}{\c_{\tau'}} > 0$. Assume that there exist infinitely many $i \geq \tau'$
  such that $\card{Q_y}{\c_i} > 0$, as we are otherwise
  done. Since there exists at least one state $q \in \act{\c_i} \cap Q_N$, rule~\majrule{13} is enabled. Therefore, there exists $\c_i'$ such
  that $\c_i \reach{*} \c_i'$ and $\card{Q_y}{\c_i} >
  \card{Q_y}{\c_i'}$. By repeatedly invoking fairness, we
  obtain some $\tau'' \geq \tau'$ such that $\card{Q_y}{\c_{\tau''}} =
  0$. This disables rules~\majrule{13} and~\majrule{14}. 
  
  Let us now consider the remaining case where $\card{Q_N}{\c_{\tau'}} = 0$. From~\Cref{lem:invariant}, it follows that for any $v\geq \tau$, $\card{Q_M}{\c_{v}} = \card{Q_m}{\c_{v}}$. Since each agent is initially in a state $q$ such that $\grp{q}$ holds, since $\card{U}{\c_\alpha} > 0$ and since $\varphi_\text{maj}(\c_0) = \false$, at least one agent has its group changed in the execution. There exists a transition $t = ((p, q), \sim, (p', q'))$ and data $d_1, d_2 \in \D$ such that $(\grp{p} \neq \grp{p'}) \lor (\grp{q} \neq \grp{q'})$ holds, $\c_{\tau-1} \reach{t} \c_\tau$ and $\c_\tau = \c_{\tau-1} - (\vec{p}_{d_1} + \vec{q}_{d_2}) + (\vec{p'}_{d_1} + \vec{q'}_{d_2})$. Observe that $t$ must be a transition arising from a rule~\majrule{2}, \majrule{3}, \majrule{6}, \majrule{7} or \majrule{8}. We identify two subcases.

  \medskip\noindent\emph{Subcase: transition $t$ arises from rule~\majrule{2},~\majrule{3},~\majrule{6} or~\majrule{7}.}  Recall that for any $v\geq\tau$, groups are stabilized in $\c_v$, and hence $\card{Q_Y}{\c_v} = \card{Q_N}{\c_v}$ since $\card{Q_M}{\c_v} = \card{Q_m}{\c_v}$. We know from the definition of the rules that either $\card{Q_N}{\c_\tau} > 0$ (which implies $\card{Q_Y}{\c_\tau} > 0$) or $\card{Q_Y}{\c_\tau} > 0$ (which implies $\card{Q_N}{\c_\tau} > 0$). Therefore, \[\c_{\tau'-1}\reach{(11)}\c_{\tau'} \text{ and  } \card{Q_Y}{\c_{\tau'-1}} = \card{Q_N}{\c_{\tau'-1}} = 1.\] By definition of rule~\majrule{12}, it follows that $\card{Q_n}{\c_{\tau'}} \geq 2$.
  
  \medskip\noindent\emph{Subcase: transition $t$ arises from rule~\majrule{8}.}
  By definition of~\majrule{8}, we have $\card{Q_n}{\c_\tau} \geq 2$. If $\card{Q_Y}{\c_\tau} = \card{Q_N}{\c_\tau} = 0$, then $\tau = \tau'$. Otherwise, we have \[\card{Q_Y}{\c_\tau} = \card{Q_N}{\c_\tau} > 0 \text{ and } \c_{\tau'-1}\reach{(11)}\c_{\tau'}.\] Again, it follows that $\card{Q_n}{\c_{\tau'}} \geq 2$.
  
  In both subcases, we have $\card{Q_n}{\c_{\tau'}} > 0$. Observe that rule~\majrule{13} is disabled. Thus, only
  rule~\majrule{14} can now change a state $q$ to a state $q'$ such
  that $\maj{q'} \neq \maj{q}$. Since $\card{Q_n}{\c_{\tau'}} > 0$,
  rule~\majrule{14} is enabled for any $j \geq \tau'$ such that
  $\card{Q_y}{\c_j} > 0$. So, by repeatedly
  invoking fairness, we obtain some $\tau'' \geq \tau'$ such that
  $\card{Q_y}{\c_{\tau''}} = 0$.
  
  Using the results shown in the proof of~\Cref{lem:convtrue}, we know there are no active states $q$ in $\c_\tau$ (and in any $\c'$ reachable from $\c_\tau$) with $\maj{q} \in \{\overline{Y}, \overline{N}\}$. Thus, we are
  done, as $\maj{q} \in \{N, n\}$ holds for every $i \geq \tau''$ and
  $q \in \act{\c_i}$.
\end{proof}

\section{Missing proofs of Section~\ref{sec:maj:io}}
\begin{applemma}[monotonicity]\label{lem:monotonicity}
  Let $\c$, $\overline{\c}$ and $\c'$ be configurations such that $\c
  \reach{*} \overline{\c}$ and $\c \embeds \c'$. There exists
  $\overline{\c'}$ such that $\c' \reach{*} \overline{\c'}$ and
  $\overline{\c} \embeds \overline{\c'}$.
\end{applemma}

\begin{proof}
  We consider the case where $\c \reach{t} \overline{\c}$. The general
  case ``${\reach{*}}$'' follows by induction. Let $\rho \colon \D \to
  \D$ be an injection such that $\c(x) \leq \c'(\rho(x))$ for all $x
  \in \D$. Let $t = ((p, q), \sim, (p', q'))$ be the transition used
  in $\c$ with $d, e \in \D$. As $\c \geq \p_d + \q_e$, we have $\c'
  \geq \p_{\rho(d)} + \q_{\rho(e)}$. So, we can use $t$ in $\c'$ to
  obtain $\overline{\c'} \defeq \c' - (\p_{\rho(d)} + \q_{\rho(e)}) +
  (\vec{p'}_{\rho(d)} + \vec{q'}_{\rho(e)})$. We have
  $\overline{\c}(d) \leq \overline{\c'}(\rho(x))$ for all $x \in \D$,
  and hence $\c' \embeds \overline{\c'}$.
\end{proof}

\begin{applemma}\label{lem:upclosed}
  The set $\U$ is upward closed and has a finite basis.
\end{applemma}

\begin{proof}\label{prf:upclosed}
  It suffices to show that $\U$ is upward closed. Indeed, since
  $\embeds$ is a well-quasi-order (\eg, see~\cite{HLLLST16}), it
  follows that $\U$ has a finite basis.

  Let $\c$ and $\c'$ be configurations such that $\c \in \U$ and $\c
  \embeds \c'$. If $O(\c') = \bot$, then we are done. Therefore,
  assume that $O(\c') = b \in \{0, 1\}$. Since $\act{\c} \subseteq
  \act{\c'}$, we have $O(\c) = b$. As $\c \in \U$, there exists a
  configuration $\overline{\c}$ such that $\c \reach{*} \overline{\c}$
  and $O(\overline{\c}) \neq b$. By \cref{lem:monotonicity}, there
  exists $\overline{\c'}$ such that $\c' \reach{*} \overline{\c'}$ and
  $\overline{\c} \embeds \overline{\c'}$. As $\act{\overline{\c}}
  \subseteq \act{\overline{\c'}}$, we either have $O(\overline{\c'}) =
  O(\overline{\c})$ or $O(\overline{\c'}) = \bot$. So, we obtain $\c'
  \in \U$.
\end{proof}

\begin{applemma}\label{lem:trunc:embeds}
  Let $k \geq 1$ and let $\c, \c'$ be configurations such that $\c
  \embeds \c'$. It is the case that $\tau_k(\c) \embeds \c$ and
  $\tau_k(\c) \embeds \tau_k(\c')$.
\end{applemma}

\begin{proof}
  For every $d \in \D$ and $q \in Q$, we have $\tau_k(\c)(d, q) =
  \min(\c(d, q), k) \leq \c(d, q)$. Thus, $\tau_k(\c) \embeds \c$
  trivially holds (using the identity function).
  
  Let $\rho \colon \D \to \D$ be an injection such that $\c(d) \leq
  \c'(\rho(d))$ for every $d \in \D$. For every $d \in \D$ and $q \in
  Q$, we have
  \[
  \tau_k(\c)(d, q) =
  \min(\c(d, q), k) \leq
  \min(\c'(\rho(d), q), k) =
  \tau_k(\c')(\rho(d), q).
  \]
  Thus, $\tau_k(\c) \embeds \tau_k(\c')$ (using injection $\rho$).
\end{proof}

\lemtruncinu*

\begin{proof}\label{prf:truncinu}
  By \Cref{lem:upclosed}, the set $\U$ has a finite basis $\{\c_1,
  \c_2, \ldots, \c_n\}$. Let $k \defeq \max\{\c_i(d, q) : d \in \D, q
  \in Q\}$. Note that $\tau_k(\c_i) = \c_i$ for all $i \in
      [1..n]$. Let us first show that $\c \in \U$ iff $\tau_k(\c) \in
      \U$.

  $\Rightarrow$) As $\c \in \U$, we have $\c_i \embeds \c$ for some $i
      \in [1..n]$. By \Cref{lem:trunc:embeds}, we have $\tau_k(\c_i)
      \embeds \tau_k(\c)$. Since $\c_i = \tau_k(\c_i)$, we obtain
      that $\c_i \embeds \tau_k(\c)$ and hence $\tau_k(\c) \in \U$.

  $\Leftarrow$) Since $\tau_k(\c) \in \U$, we have $\c_i \embeds
  \tau_k(\c)$ for some $i \in [1..n]$. By \Cref{lem:trunc:embeds}, we
  have $\tau_k(\c) \embeds \c$. Consequently, $\c_i \embeds \c$, which
  implies $\c \in \U$.

  It remains to prove the main claim. As $\S_0 \cup \S_1$ is the
  complement of $\U$, we have $\c \in \S_0 \cup \S_1$ iff $\tau_k(\c) \in
  \S_0 \cup \S_1$. Moreover, $O(\c) = O(\tau_k(\c))$ since $\act{\c} =
  \act{\tau_k(\c)}$. Therefore, $\c \in \S_b$ iff $\tau_k(\c) \in \S_b$.
\end{proof}

\lemstatelimit*

\begin{proof}\label{prf:statelimit}
  Let $\c_0\c_1\ldots\c_m$ be a fair execution reaching a stable configuration $\c_m$ with output $b\in\{0, 1\}$. Let $s_i(j)$ be the state of the $j$-th agent of datum $d$ in configuration $\c_i$. Let $a$ be such that $s_0(a) = q$ and $\c_m(d, s_m(a)) \geq k$, that is an agent initially in state $q$ that ends up in state $s_m(a) = q'$ with at least $k$ agents in $\c_m$. Note that since $\c_0(d, q) \geq \ell = |Q|\cdot (k-1) +1$, by the pigeonhole principle, such a state $q'$ must exist. We first show by induction that $\c_0 + \vec{q}_d \reach{*} \c_i + \vec{s_i(a)}_d$.

  \medskip
  \noindent\emph{Case~$i=0$}. Since $s_0(a) = q$ and $a$ is of datum $d$, it follows that $\c_0 + \vec{q}_d = \c_0 + \vec{s_0(a)}_d$, which implies that $\c_0 + \vec{q}_d \reach{*} \c_0 + \vec{s_0(a)}_d$.

  \medskip
  \noindent\emph{Case~$i>0$}. Assume $\c_0 + \vec{q}_d \reach{*} \c_{i-1} + \vec{s_{i-1}(a)}_d$. If $s_{i-1}(a) = s_i(a)$, then
  \begin{align*}
    \c_0 + \vec{q}_d &\reach{*} \c_{i-1} + \vec{s_{i-1}(a)}_d && \text{(by induction hypothesis)}\\
    &\reach{} \c_i + \vec{s_{i-1}(a)}_d && (\c_{i-1}\reach{}\c_i)\\
    &= \c_i + \vec{s_i(a)}_d && (s_{i-1}(a) = s_i(a))
  \end{align*}
  If $s_{i-1}(a) \neq s_i(a)$, there exist some $p\in Q$ and some transition $t=((p, s_{i-1}(a)), \sim, (p, s_i(a)))$ such that $\c_{i-1} \reach{t} \c_i$.
  Since the protocol is immediate observation, this transition leaves the observed agent intact in $\c_i$, that is $\c_i\geq \vec{p}_{d'}$ for some $d'\in\D$. This agent can again be observed by the agent of datum $d$ in state $s_{i-1}(a)$. Therefore,
  \begin{align*}
    \c_0 + \vec{q_d} &\reach{*} \c_{i-1} + \vec{s_{i-1}(a)}_d && \text{(by induction hypothesis)}\\
    &\reach{t} \c_i + \vec{s_{i-1}(a)}_d && (\c_{i-1}\reach{}\c_i)\\
    &\reach{t} \c_i + \vec{s_i(a)}_d && \text{(by observing agent of $d'$ in state $p$)}
  \end{align*}
  Now, assume that $\c_m + \vec{q'}_d \in \U$. We have
  \begin{align*}
    \c_m + \vec{q'}_d \in \U &\implies \tau_k(\c_m + \vec{q'}_d) \in \U && \text{(by \Cref{lem:truncinu})}\\
    &\implies \tau_k(\c_m) \in \U && (\c_m(d, q') \geq k)\\
    &\implies \c_m \in \U && \text{(by \Cref{lem:truncinu})}
  \end{align*}
  This is a contradiction to our assumption that $\c_m\in \S$. Therefore, $\c_m + \vec{q'}_d$ must be stable. Observe that since $\c_m(q') > 0$, it follows that $O(\c_m + \vec{q'}_d) = b = O(\c_m)$. This implies that an execution starting from $\c_0 + \vec{q}_d$ reaches a stable configuration with output $b$. Therefore, $\psi(\c_0) = \psi(\c_0 + \vec{q}_d)$.
\end{proof}

\lemnewcolor*

\begin{proof}\label{prf:newcolor}
  Let $d, d'$ be as described and let $\c = \c_0 \reach{t_1} \c_1
  \reach{t_2} \cdots \reach{t_i} \c_i = \c'$. We prove the claim by
  induction on $i$. As the claim trivially holds for $i = 0$, let us
  consider the case where $n > i$. We make a case distinction.

  \medskip
  \noindent\emph{Case $\c_i(d) = \c_{i-1}(d)$}. We have
  \begin{alignat*}{3}
    \c_0 + (\c_0(d))_{d'}\
    &\reach{*}\ && \c_{i-1} + (\c_{i-1}(d))_{d'}\
    && \quad\text{(by induction hypothesis)} \\
    &=\ && \c_{i-1} + (\c_i(d))_{d'} \\
    & \reach{}\ && \c_i + (\c_i(d))_{d'}
    && \quad\text{(by $\c_{i-1} \reach{} \c_i$)}.
  \end{alignat*}

  \medskip
  \noindent\emph{Case $\c_i(d) \neq \c_{i-1}(d)$.} We show that
  $(\c_{i-1}(d))_{d'} \reach{t} (\c_i(d))_{d'}$. By assumption, there
  is a transition $t = ((p, q), \sim, (p, q')) \in \delta$ such that
  $\c_{i-1} \reach{t} \c_i$, $q \in \act{\c_{i-1}(d)}$ and $q \neq
  q'$. The latter implies $(\c_{i-1}(d))_{d'} \geq \mat{q_{d'}}$. Let
  us make a case distinction on ${\sim}$.
  \begin{itemize}
  \item Assume that ${\sim}$ is ``$=$''. If $p=q$, we have $\c_{i-1}(d, q) \geq 2$. Thus, we have $(\c_{i-1}(d))_{d'}\reach{t} (\c_i(d))_{d'}$. Otherwise, we have $\c_{i-1}(d, p) \geq 1$ and $\c_{i-1}(d, q) \geq 1$. This also implies that $(\c_{i-1}(d))_{d'} \reach{t}
    (\c_i(d))_{d'}$ holds.

  \item Assume that $\sim$ is ``$\neq$''. Some color $d'' \in
    \supp{\c_0}$ satisfies $p \in \act{\c_{i-1}(d'')}$ and hence $p
    \in \act{\c_i(d'')}$. Observe that $d'' \in \supp{\c_0}$ implies
    that $d'' \neq d'$. Therefore, we have
    $(\c_{i-1}(d))_{d'} \reach{t} (\c_i(d))_{d'}$. \qedhere
  \end{itemize}
\end{proof}

\lemformthreshold*

\begin{proof}\label{prf:formthreshold}
  Let $\f$ be a form with $\act{\f} \subseteq I$, let $\{\c_1, \ldots,
  \c_n\}$ be a basis of $\U$ given by \Cref{lem:upclosed}, and let $m
  \defeq \max(|\supp{\c_1}|, |\supp{\c_2}|, \ldots,
  |\supp{\c_n}|)$. We define $h(\f)$ as follows:
  \[
  h(\f) \defeq (m-1) \cdot |Q|^{\sum_{q \in Q} \vec{f}(q)} + 1.
  \]

  Let $\c_0$ be an initial configuration such that $\#_{\f}(\c_0) \geq
  h(\f)$. Let $b \defeq \psi(\c_0)$. Let $\c \in \S_b$ be such that
  $\c_0 \reach{*} \c$. By the pigeonhole principle, there exists $d
  \in \supp{\c_0}$ such that
  \begin{itemize}
  \item $\c_0(d) = \f$, and

  \item $\c(d) = \f'$ with $\#_{\f'}(\c) \geq m$.
  \end{itemize}
  Let us now inject the form $\f$ with a new datum $d' \in \D
  \setminus \supp{\c_0}$. By Lemma~\ref{lem:newcolor}, we have
  \[
  \c_0 + \f_{d'} =
  \c_0 + (\c_0(d))_{d'} \reach{*}
  \c + (\c(d))_{d'} =
  \c + \f'_{d'}.
  \]
  Let $\c' \defeq \c + \f'_{d'}$. As $\c(d) = \f'$ and $\c \in \S_b$,
  we have $O(\c) = O(\f'_{d'}) = b$ and so $O(\c') = b$.

  Let us now show that $\c' \in \S_b$. For the sake of contradiction,
  suppose that $\c' \in \U$. There exists $i \in [1..n]$ such that
  $\c_i \embeds \c'$. By the latter, there exists an injection $\rho
  \colon \D \to \D$ such that $\c_i(x) \leq \c'(\rho(x))$ for every $x
  \in \D$. Since 
  \[
    \#_{\f'}(\c') > \#_{\f'}(\c) \geq m \geq |\supp{\c_i}|,
  \]
  $d^* \in \D$ such that $\c'(d^*) = \f'$, and either
  $d^*$ does not belong to the image of $\rho$ or
  $\c_i(\rho^{-1}(d^*)) = \vec{0}$. Let us show that $\c_i \embeds
  \c$, which yields the contradiction $\c \in \U$. We make a case
  distinction.

  \medskip
  \noindent\emph{Case $\rho^{-1}(d')$ undefined}. We trivially have
  $\c_i(x) \leq \c'(\rho(x)) = \c(\rho(x))$ for all $x \in \D$.

  \medskip
  \noindent\emph{Case $d^* = d'$}. Since $\c_i(\rho^{-1}(d')) =
  \c_i(\rho^{-1}(d^*)) = \vec{0} = \c(d') = \c(d^*)$, and as $\c$ and
  $\c'$ only differ on $d'$, we have $\c_i(x) \leq \c(\rho(x))$ for
  every $x \in \D$. Thus, $\c_i \embeds \c$.
  
  \medskip
  \noindent\emph{Case $d^* \neq d'$}. Let $\rho' \colon \D \to \D$ be
  the injection given by $\rho$, but with $d^*$ playing the role of
  $d'$ (and possibly the converse to ensure that $\rho'$ is indeed an
  injection), \ie\
  \[
  \rho'(x) \defeq
  \begin{cases}
    d'        & \text{if } \rho^{-1}(d^*) \text{ is defined and }
                x = \rho^{-1}(d^*), \\
    d^*       & \text{if } x = \rho^{-1}(d'), \\
    \rho(x)   & \text{otherwise}.
  \end{cases}
  \]
  As $d^* \neq d'$, we have $\c(d^*) = \c'(d^*) = \vec{f}'$. Thus,
  \[
  \c_i(\rho^{-1}(d')) \leq
  \c'(d') =
  \vec{f}' =
  \c(d^*) =
  \c(\rho'( \rho^{-1}(d') )).
  \]
  Furthermore, if $\rho^{-1}(d^*)$ is defined, then we have
  $\c_i(\rho^{-1}(d^*)) = \vec{0} = \c(d') =
  \c(\rho'(\rho^{-1}(d^*)))$. Thus, since $\c$ and $\c'$ only differ
  on datum $d'$, we have $\c_i(x) \leq \c(\rho'(x))$ for all $x \in
  \D$, which implies $\c_i \embeds \c$.
\end{proof}

\clearpage
\section{Missing proofs of Section~\ref{sec:interval}}
Let us first recall all of the transitions in a single table:

\newcommand{\colorpred}{\emph{color pre.}}

\begin{center}
  \resizebox{\columnwidth}{!}{%
  \begin{tabular}{cccc}
    \bottomrule
    \multicolumn{4}{c}{\textbf{Leader and controller election}} \\
    \toprule
    \emph{rule}
    & \emph{state precondition}
    & \colorpred
    & \emph{state update} \\

    \midrule
    \multirow{2}{*}{(1)}
    & $\lead{p}$
    & \multirow{2}{*}{$d_1 = d_2$}
    & $\role{q'} = -|\role{q}|$ \\
    & $\lead{q}$
    &
    & $\lnot\lead{q'}$ \\

    \midrule
    \multirow{2}{*}{(2)}
    & $\ctrl{p} = 1$
    & \multirow{2}{*}{\none}
    & \\
    & $\ctrl{q} = 1$
    &
    & $\ctrl{q'} = -1$ \\

    \bottomrule

    \multicolumn{4}{c}{\textbf{Element count by datum}} \\
    \toprule
    \emph{rule}
    & \emph{state precondition}
    & \colorpred
    & \emph{state update} \\
    \midrule
    \multirow{2}{*}{(3)}
    & $\mathmakebox[35pt][r]{\init{p} =}~\mathmakebox[35pt][l]{\init{q}}$
    & \multirow{2}{*}{$d_1 = d_2$}
    & \\
    & $\mathmakebox[35pt][r]{\val{q} =}~\mathmakebox[35pt][l]{\val{p} < r}$
    &
    & $\val{q'} = \val{q} + 1$ \\

    \midrule
    \multirow{3}{*}{(4)*}
    & $\mathmakebox[60pt][r]{\cnt{\init{p}}{q}}\mathmakebox[30pt][l]{~< \val{p}}$
    & \multirow{3}{*}{$d_1 = d_2$}
    & $\cnt{\init{p}}{q'} = \val{p}$\\
    & $\mathmakebox[60pt][r]{\lead{q}}\mathmakebox[30pt][l]{}$
    &
    & $\texttt{if } (\role{q} > 0 \land {}$ \\

    &
    &
    & \hspace{1.7cm}$\val{p} \notin T(\role{q}, \init{p}))$\texttt{:} \\
    &
    &
    & $\quad \role{q'} = -\role{q}$ \\

    \bottomrule

    \multicolumn{4}{c}{\textbf{Role distribution and task tracking}} \\
    \toprule
    \emph{rule}
    & \emph{state precondition}
    & \colorpred
    & \emph{state update} \\

    \midrule
    \multirow{3}{*}{(5)}
    & $\lead{q}$
    & \multirow{3}{*}{\none}
    & \\
    & $\role{q} = 0$
    &
    & $\role{q'} = i$ \\
    & $\bigwedge_{j \in [1..m]} \cnt{j}{q} \in T(i, j)$
    &
    & \\

    \midrule
    \multirow{5}{*}{(6)*}
    & $\ctrl{p}$
    & \multirow{5}{*}{\none}
    & \\
    & $\lead{q}$
    &
    & \\
    & $\role{q} = i > 0$
    & 
    & $\role{q'} = -i$ \\
    & $\bigvee_{i' \in [1..n] \setminus \{i\}} \big(\lnot\task{i'}{p} \land \hspace{44pt}$
    &
    & \\
    & $\hspace{43pt}\bigwedge_{j \in [1..m]} \cnt{j}{q} \in T(i', j)\big)$
    &
    &  \\

    \midrule \\[-16.5pt]
    \midrule
    \multirow{2}{*}{(7)*}
    & $\role{p} \neq 0$
    & \multirow{2}{*}{\none}
    & $\task{|\role{p}|}{q'} = (\role{p} > 0)$ \\
    & $\ctrl{q} = 1$
    &
    & \\

    \midrule
    \multirow{3}{*}{(8)*}
    & $\ctrl{p} = 1$
    & \multirow{3}{*}{\none}
    & \\
    & $\role{q} < 0$
    &
    & $\role{q'} = 0$ \\
    & $\lnot\task{|\role{q}|}{p}$
    &
    & \\
    \midrule \\[-16.5pt]
    \midrule
    \multirow{2}{*}{(9)}
    & \makebox[30pt][l]{$\ctrl{p} = -1$}
    & \multirow{2}{*}{\none}
    & \\
    & \makebox[30pt][l]{$\ctrl{q} = 1$}
    &
    & $\bigwedge_{i \in [1..m]} \lnot\task{i}{q'}$ \\
    
    \midrule
    \multirow{3}{*}{(10)}
    & \makebox[30pt][l]{$\ctrl{p} = 1$}
    & \multirow{3}{*}{\none}
    & \\
    & \makebox[30pt][l]{$\ctrl{q} = -1$}
    &
    & $\ctrl{q'} = 0$ \\
    & $\bigwedge_{i \in[1..m]} \lnot\task{i}{p}$
    &
    & \\
    
    \bottomrule

    \multicolumn{4}{c}{\textbf{Output propagation}} \\
    \toprule
    \emph{rule}
    & \emph{state precondition}
    & \colorpred
    & \emph{state update} \\

    \midrule
    \multirow{1}{*}{(11)*}
    & $\ctrl{p}$
    & \multirow{1}{*}{\none}
    & $\out{q'} = \bigwedge_{i=1}^n \task{i}{p}$ \\
    \bottomrule
  \end{tabular}}
\end{center}
Note that for the sake of brevity, for the remainder of the section, we use the notation $\c\reach{(n)}\c'$ to denote that $\c\reach{t}\c'$ using a transition $t$ arising from rule $(n)$.
\lemldrctrl*

\begin{proof}\label{prf:ldr:ctrl}
  We prove the two parts of the claim independently. The validity of
  the claim follows by taking the maximum among the two thresholds
  $\tau$.

  \medskip
  \noindent\emph{Threshold for $Q_L$}. Note that $\c \reach{(1)} \c'$
  implies $\card{Q_L}{\c'} = \card{Q_L}{\c} - 1 > 0$, and $\card{Q_L}{\c} = \card{Q_L}{\c'}$
  otherwise. Therefore, there exist $\tau \in \N$ and $k \geq 1$ such
  that $\card{Q_L}{\c_\tau} = \card{Q_L}{\c_{\tau+1}} = \cdots =
  k$. If $k = 1$, then we are done. Thus, for the sake of
  contradiction, assume that $k \geq 2$. Rule~\intrule{1} is enabled
  in each of $\c_\tau, \c_{\tau+1}, \ldots$. By fairness, this implies
  that for some $\tau' \geq \tau$, we have $\card{Q_L}{\c_{\tau'}} >
  \card{Q_L}{\c_{\tau'+1}}$, which is a contradiction.

  \medskip
  \noindent\emph{Threshold for $Q_C$}. The proof is the same, but
  using rule~\intrule{2}.
\end{proof}

\lemcount*

\begin{proof}\label{prf:count}
  Note that rule~\intrule{3} is the only rule modifying the value of a state; it requires two states with the same value and it increases the value of a single state by $1$. 
  Therefore, for any $\c, \c'$ and $d\in\D$ such that $\c
  \reach{} \c'$ with $\c' = \c - \q_d + \vec{q'}_d$ for some states $q,
  q'\in Q$, either $\val{q'} = \val{q}$, or $\val{q'} = \val{q} + 1$
  holds. Further
  observe that no rule changes a state $q$ to a state $q'$ such that
  $\init{q'} \neq \init{q}$, \ie $\init{q}$ is immutable for any $q$.

  Let $d \in \D$ and $j \in [1..m]$ be such that $\c_0(d, Q_j) >
  0$. We define
  \[
  Q_{j, k} \defeq \{q\in Q : \init{q} = j \land \val{q} = k\}
  \text{ and }
  Q_{j, >k} \defeq \bigcup_{k'= k+1}^r Q_{j, k'}.
  \]
  We claim that, for every $1 \leq k < \min(\c_0(d, Q_j) + 1, r)$,
  there exists $\tau \in \N$ such that, for every $\tau' \geq \tau$,
  the following holds:
  \[
  \card{Q_{j, k}}{\c_{\tau'}(d)} = 1
  \text{ and }
  \card{Q_{j, >k}}{\c_{\tau'}(d)} = \c_0(d, Q_j) - k.
  \]
  Before proving the claim, let us see how it allows to conclude. If
  $\c_0(d, Q_j) \leq r$, then this confirms the existence of an agent of $d$
  in a $(d, j)$-valid state. Note that if $k = r - 1$, then $Q_{j,
    >k}$ is in fact $Q_{j, r}$. Therefore, if $\c_0(d, Q_j) > r$, then
  agents will accumulate in some states of $Q_{j, r}$. In that case,
  there will be more than one agent of $d$ in a $(d, j)$-valid state. Let us now prove the claim by induction on $k$.

  \medskip
  \noindent\emph{Base case $(k = 1)$}. Since all states $q\in I$ satisfy $\val{q} = 1$, we gather from the previous observations that for $\c_i, \c'$ such that $\c_i \reach{*} \c'$, $\c_i(d, Q_{j,1}) \geq \c'(d,Q_{j,1})$ holds.  Note that the fact that two states with the same value are needed for rule~\intrule{3} to decrease $\c'(d, Q_{j,1})$ implies that $\c'(d, Q_{j, 1}) \geq 1$. Note also that for any $q\in\act{\c_i}$, either $\val{q} = 1$ or $\val{q} > 1$. By the above observations,
  there exist $\ell \in \N$ and $\kappa \geq 1$ such that $\card{Q_{j, 1}}{\c_\ell(d)} = \card{Q_{j,
      1}}{\c_{\ell+1}(d)} = \cdots = \kappa$. If $\kappa = 1$, then $\card{Q_{j, >1}}{\c_{\ell}(d)} = \c_0(d, Q_j) - 1$ and we are
  done. Otherwise, rule~\intrule{3} is enabled in each of $\c_\ell,
  \c_{\ell+1}, \ldots$. By fairness, this implies the existence of
  $\ell' \geq \ell$ such that
  \[
  \card{Q_{j, 1}}{\c_{\ell'}(d)} > \card{Q_{j, 1}}{\c_{\ell'+1}(d)}
  \text{ and }
  \card{Q_{j, >1}}{\c_{\ell'}(d)} < \card{Q_{j, >1}}{\c_{\ell'+1}(d)},
  \]
  which is a contradiction.

  \medskip
  \noindent\emph{Induction step}. Assume that the claim is true for
  $k$ such that $1 < k < r - 1$ and $k \leq \c_0(d, Q_j)$. Let us show
  that the claim holds for $k + 1$. If $k = \c_0(d, Q_j)$, then $k + 1
  \not< \min(\c_0(d, Q_j) + 1, r)$, and hence we are trivially done as
  the claim does not have to hold for $k + 1$. Thus, let us assume
  that $k < \c_0(d, Q_j)$. By induction hypothesis, there exists $v
  \in \N$ such that, for every $v' \geq v$, the following holds:
  \[
  \card{Q_{j, k}}{\c_{v'}(d)} = 1
  \text{ and }
  \card{Q_{j, >k}}{\c_{v'}(d)} = \c_0(d, Q_j) - k.
  \]
  Since values are incremented by $1$ (\ie there are no gaps in
  values), at least one agent is in a state of $Q_{j, k + 1}$. Suppose
  that for some $w \geq v$, there exist states $p, q \in Q_{j, k+1}$
  such that $\c_w \geq \p_d + \q_d$. Then, rule~\intrule{3}
  is enabled, implying the existence of a configuration $\c'$ such
  that $\c_w \reach{*} \c'$,
  \[
  \card{Q_{j, k+1}}{\c_w(d)} > \card{Q_{j, k+1}}{\c'(d)}
  \text{ and }
  \card{Q_{j, >k+1}}{\c_w(d)} < \card{Q_{j, >k+1}}{\c'(d)}.
  \]
  By fairness, this implies the existence of some $w'$ with
  $\card{Q_{j, k+1}}{\c_{w'}(d)} = 1$. We are done since
  \begin{align*}
    \card{Q_{j, >k+1}}{\c_{w'}(d)}
    &= \card{Q_{j, >k}}{\c_{w'}(d)} - \card{Q_{j, k+1}}{\c_{w'}(d)}
    && \text{(by definition of $Q_{j, \cdot}$)} \\
    &= \card{Q_{j, >k}}{\c_{w'}(d)} - 1 \\
    &= (\c_0(d, Q_j) - k) - 1 \\
    &= \c_0(d, Q_j) - (k + 1). && \qedhere
  \end{align*}
\end{proof}

\lemldrcnt*

\begin{proof}\label{prf:ldrcnt}
  Note that only rule~\intrule{4} changes a state $q$ to a state $q'$
  such that for some $j \in [1..m]$, $\cnt{j}{q'} \neq
  \cnt{j}{q}$. Moreover, using rule~\intrule{4} leads to $\cnt{j}{q'}
  > \cnt{j}{q}$.
  
  Let $d \in \D$ and $j \in [1..m]$. If $\c_0(d, Q_j) = 0$, then the
  leader's count for that element is never updated, \ie it never meets
  a state $q$ of its datum such that $\init{q} = j$ and therefore its
  count is always $0$ for element $j$. Otherwise, let $\tau$ be the
  threshold given by \Cref{lem:count}, let $v \geq \max(\alpha, \tau)$
  and let $q \in \act{\c_v(d)}$ be a $(d, j)$-valid state given by
  \Cref{lem:count}. If there exists $p \in \act{\c_v(d)} \cap Q_L$
  such that $\cnt{j}{p} < \val{q}$, then, rule~\intrule{4} is
  enabled. This implies the existence of a configuration $\c'$ such
  that $\c_v \reach{*} \c'$ and, for any $q' \in \act{\c'(d)} \cap
  Q_L$, $\cnt{j}{q'} = \min(\c_0(d, Q_j), r)$ holds. Therefore, by
  fairness, there exists $v'$ such that
  \[
  \forall q' \in \act{\c_{v'}(d)} \cap Q_L, j \in [1..m] :
  \cnt{j}{q'} = \min(\c_0(d, Q_j), r).
  \]
  Note that the precondition of rule~\intrule{4} is no longer met once
  the leader has the correct count since no other state can have a
  higher count for $j$ in $d$. Since this is valid for any datum and
  any $j$, there exists a configuration such that for any datum, the
  datum's leader has the correct counts for all elements.
\end{proof}

\lemcorrectrole*

\begin{proof}\label{prf:correctrole}
  Initially, each state $q \in \act{\c_0}$ satisfies $\role{q} = 0$.
  Observe that only rule~\intrule{5} can assign a strictly positive
  role to a state. Formally, for a datum $d \in \D$, a role $i \in
  [1..n]$ and configurations $\c_\tau, \c'$ such that $\c_\tau
  \reach{t} \c'$ and $\c' = \c_\tau - \q_d + \vec{q'}_d$ for some $q,
  q'\in Q$ and transition $t$ arising from rule~\intrule{5}, the
  leader's state $q$ must satisfy $\cnt{j}{q}\in T(i, j)$ for all $j
  \in [1..m]$. Thus, a leader must have the correct counts for a role
  before self-assigning the aforementioned role.

  Futher observe that rule~\intrule{4} is the only rule changing a
  count for a leader. Formally, for a datum $d \in \D$, a role $i \in
  [1..n]$ and configurations $\c_\tau, \c'$ such that $\c_\tau
  \reach{t} \c'$ and $\c' = \c_\tau - \q_d + \vec{q'}_d$ for some $q,
  q'\in Q$ and transition $t$ arising from rule~\intrule{4}, this
  leader must enter a new state $q'$ such that $\role{q'} < 0$ if it
  no longer satisfies the count condition for the corresponding
  interval. Therefore, after updating its count, the leader either has
  a strictly negative assigned role or it still satisfies the role it
  previously assumed.
  
  Therefore, for any $\tau \in \N$, all leaders with strictly positive
  roles have the correct counts for their role in $\c_\tau$.
\end{proof}

\begin{applemma}\label{lem:nonldrroles}
  For every $\tau \in \N$ and
  $q \in \act{\c_\tau} \setminus Q_L$, it is the case that
  $\role{q} \leq 0$.
\end{applemma}

\begin{proof}
  Observe that only rule~\intrule{5} can assign a strictly positive
  role to a state. Formally, for a datum $d \in \D$ and configurations
  $\c_\tau, \c'$ such that $\c_\tau \reach{(5)} \c'$ and $\c' = \c_\tau
  - \q_d + \vec{q'}_d$ for some $q, q'\in Q$, state $q$ must satisfy
  $\lead{q}$. Therefore, non-leaders are never directly assigned a
  strictly positive role.
  
  Further note that only rule~\intrule{1} removes leadership from a
  state. Formally, for a datum $d \in \D$ and configurations $\c_\tau,
  \c'$ such that $\c_\tau \reach{(1)} \c'$ and $\c' = \c_\tau - \q_d +
  \vec{q'}_d$ for $q, q'\in Q$,
  state $q'$ must satisfy $\role{q'} < 0$. Hence, non-leaders are
  assigned a strictly negative role when their leadership is stripped
  away. Note that leadership removal happens only once.

  Therefore, non-leaders can never have a strictly positive assigned
  role.
\end{proof}

\lemnonldrstable*

\begin{proof}\label{prf:nonldrstable}
  From \Cref{lem:nonldrroles}, we know non-leaders cannot have a
  strictly positive assigned role. Let $Q_0\defeq\{q\in Q\colon
  \lnot\lead{q} \land \role{q} = 0\}$. Note that only rule~\intrule{8}
  changes the role of a non-leader, and it changes it to $0$. This,
  along with the fact that leadership removal is definitive, implies
  that $\card{Q_0}{\c_j} \geq \card{Q_0}{\c_i}$ holds for every $j >
  i$.

  Assume that for some $v \in \N$, there
  is some state $q\in\act{\c_v}$ such that $\lnot\lead{q}$ holds and
  $\role{q} < 0$. Let $i\defeq |\role{q}|$. Then, rule~\intrule{7} is enabled in
  $\c_v$. This means that there exists a configuration $\c'$ and a
  state $q \in \act{\c'} \cap Q_C$ such that $\c_v \reach{(7)} \c'$ and
  $\task{i}{q} = \false$. Rule~\intrule{8} is enabled in
  $\c'$.
  This implies the existence of a configuration $\c''$ such that
  $\c'\reach{*}\c''$ and $\card{Q_0}{\c_v} < \card{Q_0}{\c''}$.
    
  By fairness, some configuration $\c_{v'}$ also satisfies this condition. Thus, for some $\tau\geq v'$, we have
  $
    \card{Q_0}{\c_\tau} = \card{Q \setminus Q_L}{\c_\tau}
  $. Since it has been established that no rule changes a non-leader's role to a non-$0$ value, this implies that for any $\tau'\geq \tau$, non-leaders are in states $q\in Q_0$.
\end{proof}

\lemuniquectrl*

\begin{proof}\label{prf:uniquectrl}
  Observe that rule~\intrule{2} is the only rule assigning a strictly
  negative controller value. We know that from $\c_\alpha$, there
  exists a unique controller. Let $Q_{-1} \defeq \{q \in Q : \ctrl{q}
  = -1\}$. Note that for any $j > i \geq \alpha$, since
  rule~\intrule{2} is disabled, we have $ \card{Q_{-1}}{\c_j} \leq
  \card{Q_{-1}}{\c_i} $.

  Assume that for some $v \geq \alpha$,
  there exists a state $q \in \act{\c_v}$ satisfying $\ctrl{q} =
  -1$. Then, we know that rule~\intrule{9} is enabled in $\c_v$,
  implying the existence of a configuration $\c'$ such that $\c_v
  \reach{(9)} \c'$ and
  \begin{align*}
  \bigwedge_{i \in[1..m]} \lnot\task{i}{q}
  &&\text{ for some } q \in \act{\c'} \cap Q_C
  \end{align*}
  This means that for $\c'$, rule~\intrule{10} is enabled. Therefore,
  there exists a configuration $\c''$ such that $\c'\reach{(10)}\c''$ and
  $\card{Q_{-1}}{\c_v} > \card{Q_{-1}}{\c''}$. Hence, by fairness,
  there exists $\tau \geq v$ such that $\card{Q_{-1}}{\c_\tau} = 0$.
  Since no rule can change a state $q$ with $\ctrl{q} = 0$ to a state
  $q'$ with $\ctrl{q} \neq 0$, this holds for any $\tau' \geq \tau$.
\end{proof}

\lemctrltask*

\begin{proof}\label{prf:ctrltask}
  Note that only rule~\intrule{7} can assign $\true$ to some task $i$
  for the controller. Let $R_i \defeq \{q\in Q : |\role{q}| = i\}$.

  Let $i \in [1..n]$ and let $q \in Q_C$ satisfy $\lnot\task{i}{q}$
  and $q'\in Q_C$ satisfy $\task{i}{q'}$. Let some $v \geq \gamma$ be such that $\c_v \reach{(7)}
  \c_{v+1}$,  $\c_{v+1} = \c_{v} - \q_d +
  \vec{q'}_d$. Rule~\intrule{7} implies the existence of some
  state $p \in \act{\c_{v+1}} \cap R_i$. Note that for $\c_\gamma$, a unique controller has been permanently elected.

  Let $w \geq v+1$. For $\c_w \reach{} \c_{w+1}$ to satisfy
  $\card{R_i}{\c_w} > \card{R_i}{\c_{w+1}}$, we must have $\c_w
  \reach{(8)} \c_{w+1}$. For that to be the case, there has to be some
  state $q'' \in \act{\c_w} \cap Q_C$ such that $\lnot\task{i}{q''}$
  holds. Thus, for any $\c_w$ for which there exists a state $q^* \in
  \act{\c_w}\cap Q_C$ satisfying $\task{i}{q^*}$, there exists a state
  $p \in \act{\c_w} \cap R_i$.
\end{proof}

\lemint*

\begin{proof}[Proof of $\Rightarrow$)]\label{prf:int}
  By definition, $\c_\beta$ has a unique controller elected and a
  unique leader for each datum, carrying the correct counts for its
  datum's initial states. Since $\psi(\c_0)$ is $\true$, a subset of
  these leaders can fulfill the $n$ roles simultaneously. While the
  roles might not be correctly distributed at $\c_\beta$, we show that
  for some configuration $\c$ reachable from $\c_\beta$, the
  controller's tasks will all be $\true$. Let us consider the
  following lemmas and their implications:
  \begin{itemize}
    \item from \Cref{lem:correctrole}, a leader with an assigned role always \emph{thinks} it can currently fill this role,
    \item from \Cref{lem:nonldrstable}, all non-leaders eventually have their roles set to $0$, leaving only the elected leaders with (possibly) assigned roles,
    \item from \Cref{lem:uniquectrl} and \Cref{lem:ctrltask}, eventually, the unique controller can only have a task $i$ set to $\true$ if a leader has its role set to $\pm i$.
  \end{itemize}
  Therefore, for configuration $\c_\gamma$, only leaders have a non-zero role. These leaders also correctly \emph{think} they can fulfill the roles to which they are assigned. Finally, there is a unique controller and no agent is in a negative controller state.
  
  We now define a \emph{valid assignment} $A$ as an injective function $A \colon [1..n] \to \D$ such that if $A(i) = d$, then $d$ matches role $i$. For a configuration $\c$  such that $\c_0\reach{*} \c$ and a valid assignment $A$, let $s_A(\c)$ be the set of correctly assigned colors in $\c$ with respect to $A$, \ie
  \[
    s_A(\c)\defeq \{d\in\D\colon (\exists i\in[1..n] : A(i) = d) \land (\exists q\in Q_L : \c(d, q) > 0 \land \role{q} = i)\}.
  \]
  Moreover, we say that a configuration $\c$ is \emph{fixed} if:
  \begin{itemize}
  \item $\c_\gamma \reach{*} \c$,

  \item there exists $q \in \act{\c}\cap Q_C$ such that
    $\bigwedge_{i=1}^{n} \task{i}{q}$, and
      
  \item $\role{q} \geq 0$ for every $q \in \act{\c}$.
  \end{itemize}
  Note that a strictly negative role can only be assigned by rule~\intrule{1} (leader election), rule~\intrule{4} (count change) and rule~\intrule{6} (self-reassignment). Observe that rules~\intrule{1},~\intrule{4} and~\intrule{6} are disabled for a fixed configuration. We also know that the controller can only change one of its tasks to $\false$ if it observes a negative role through rule~\intrule{7} or a negative controller through rule~\intrule{9}. Since a fixed configuration does not have a strictly negative role, nor can it generate one, rule~\intrule{7} cannot change any of the controller's tasks to $\false$. Furthermore, rule~\intrule{9} is disabled for a fixed configuration. Hence, a fixed configuration can only reach fixed configurations.

  Now, let $A$ be a valid assignment. For any $v \geq \gamma$, we observe two cases.

  \medskip\noindent\emph{Case: $\exists A' : |s_{A'}(\c_v)| = n$.} There exists a fixed configuration $\c'$, reachable from $\c_v$. 

  \medskip\noindent\emph{Case: $\forall A' : |s_{A'}(\c_v)| < n$.} This implies that some role $i$ is not taken by any leader. If $A(i)$'s leader has no assigned role, then by rule~\intrule{5} there exists a configuration $\c'$, reachable from $\c_v$, such that $|s_A(\c_v)| < |s_A(\c')|$ and
    $
      \role{q} = i \text{ for some } q \in \act{\c'(A(i))} \cap Q_L
    $.
      
    Otherwise, we observe two cases:
    \begin{itemize}
      \item The controller has its task $i$ set to $\false$, allowing self-reassignment of $A(i)$'s leader through rule~\intrule{6}.
      \item The controller has its task $i$ set to $\true$. Since role $i$ is not taken by any leader, by \Cref{lem:ctrltask}, for the controller to have its task $i$ set to $\true$, some $q\in\act{\c_v}\cap Q_L$ satisfies $\role{q} = -i$. Thus, by rule~\intrule{7}, there exists a configuration $\c'$ reachable from $\c_v$ such that
      $
        \lnot\task{i}{q}
      $ for some $q \in \act{\c'} \cap Q_C$.
      The latter allows self-reassignment of $A(i)$'s leader through rule~\intrule{6}.
    \end{itemize}
    Consequently, there also exists a configuration $\c''$, reachable
    from $\c'$, such that $|s_A(\c_v)| < |s_A(\c'')|$ and $\role{q} =
    i$ for some $q \in \act{\c''(A(i))} \cap Q_L$.

  By fairness, since there exists a finite number of configurations and a finite number of valid assignments for $\supp{\c_0}$, either for some $v$, $|s_A(\c_v)| = n$ holds and, by fairness, $\c_v$ leads to a fixed configuration $\c_\tau$; or for some $v$ and some valid assignment $A'\neq A$, $|s_{A'}(\c_v)| = n$ holds, and again, by fairness, this leads to a fixed configuration $\c_\tau$. Since we know that fixed configurations can only reach fixed configurations and that they also guarantee the controller has its task list completely assigned, this is true for any $\tau'\geq \tau$.
\end{proof}

\begin{proof}[Proof of $\Leftarrow$)]\label{prf:intfalse}
  We prove the contrapositive. Assume $\psi(\c_0) = \false$. Let $W(\c)$ be the set of data $d \in \supp{\c_0}$ such that for some $q\in\act{\c(d)}\cap Q_L$, 
  \begin{itemize}
    \item $|\role{q}| = i$ holds for some $i\in[1..n]$, and
    \item $\c(d, Q_j) \notin T(i, j)$ holds for some $j\in[1..m]$.  
  \end{itemize}
    Informally, $W(\c)$ is the set of data in $\c$ with a leader having a wrongly assigned role (positive or negative), \ie, a role the data cannot match. Since $\card{Q_L}{\c_\beta} = |\supp{\c_0}|$ and the number of leaders never increases, for any $\c'$ such that $\c_\beta\reach{*} \c'$, we have $W(\c') \leq |\supp{\c_0}|$. 
  Furthermore, by \Cref{lem:ldrcnt}, the counts are stable and correct for leaders past $\c_\beta$. The latter implies that for any $t \geq \beta$ and datum $d$, if a role $i$ is assigned by rule~\intrule{5} to $d$'s leader, then $d$ matches role $i$. Therefore, for any $i,j \geq \beta$ with $i \leq j$, we have $W(\c_j) \leq W(\c_i)$.

  Assume that there exists some $\pi \geq \beta$ such that $W(\c_\pi) > 0$. Then, for some $d \in \supp{\c_0}$, there exists a state $q\in\act{\c_\pi(d)}\cap Q_L$ that satisfies $|\role{q}| = i$, where $i$ is a role $d$ does not match. Since $\pi \geq \beta$, by \Cref{lem:ldrcnt}, state $q$ must satisfy $\bigwedge_{j \in [1..m]} \cnt{j}{q} = \min(\c_0(d, Q_j), r)$, \ie, it must have the right counts for each elements for $d$.

  Thus, by \Cref{lem:correctrole}, $\role{q}$ must be strictly negative. This means that rule~\intrule{7} is enabled in $\c_\pi$. So, there exist $\c'$ and $q \in \act{\c'} \cap Q_C$ such that $\c_\pi \reach{*} \c'$ and $\lnot\task{i}{q}$ holds. Note that rule~\intrule{8} is now enabled in $\c'$. Therefore, there exists a configuration $\c''$ such that $\c'\reach{(8)}\c''$ and $W(\c'') < W(\c_\pi)$. Thus, by fairness, there exists $u \geq \pi$ such that $W(\c_u) = 0$.

  From \Cref{lem:ctrltask}, for any $v\geq\beta$, for all $q\in\act{\c_v}\cap Q_C$, if $\task{i}{q}$ holds, then there exists some state $q'\in\act{\c_v}\cap Q_L$ such that $\role{q'} = \pm i$. Recall from \Cref{lem:nonldrstable} that any non-leader eventually has a role of $0$. 
  Furthermore, observe that
  \begin{itemize}
    \item for $\c_u$ described above, for any $d\in\D$, if $d$'s leader has an assigned role $i$, then $d$ matches role $i$; and
    
    \item for every $w \geq \max (u, \gamma)$, since $\psi(\c_0)$ is $\false$, there is no valid assignment $A$ such that $|s_A(\c_w)| = n$.
  \end{itemize}
   Thus, no $\c_w$ satisfies the following for all $i \in [1..n]$:
  \[
    \exists d\in\supp{\c_0}, q\in\act{\c_w(d)}\cap Q_L : |\role{q}| = i.
    \]    
  Therefore, there exists $\tau \in \N$ such that, for any $\c'$ with $\c_\tau\reach{*} \c'$, there exists $i\in[1..n]$ such that, for every $q\in\act{\c'}\cap Q_C$, it is the case that $\lnot\task{i}{q}$ holds.
\end{proof}

\thminterval*

\begin{proof}\label{prf:interval}
  \medskip
  \noindent\emph{Case: $\psi(\c_0) = \true$.} Let $Q_{\true}\defeq \{q\in Q\colon \out{q}\}$. Note that only one controller exists for any configuration reachable from $\c_\alpha$. Recall the notion of fixed configurations from the proof of~\Cref{lem:int}. For a fixed configuration $\c$, if there exists a state $q \in \act{\c}$ such that $\lnot\out{q}$ holds, since the controller has its task list completely assigned, then there exists a configuration $\c'$ reachable from $\c$ such that $\card{Q_{\true}}{\c} < \card{Q_{\true}}{\c'}$. Note that the controller never propagates the $\false$ output in a fixed configuration; it follows that $\card{Q_{\true}}{\c_v} \leq \card{Q_{\true}}{\c^*}$ for any $\c^*$ reachable from $\c_v$. Invoking fairness, this implies that there also is a configuration $\c''$ reachable from $\c'$ such that $\card{Q_{\true}}{\c''} = |\c''|$. From the proof of~\Cref{lem:int}, we also know that if $\psi(\c_0) = \true$, then a fixed configuration is eventually reached. Therefore, some configuration $\c_\tau$ such that $O(\c_\tau) = \true$ is reached; it remains $\true$ for any configuration $\c_{\tau'}$ such that $\c_\tau \reach{*} \c_{\tau'}$.

  \medskip
  \noindent\emph{Case: $\psi(\c_0) = \false$.} From~\Cref{lem:int}, we know that for some $u\in\N$ and for any $v\geq u$, at least one task is unassigned for the controller in $\c_v$. Assume that there exists some $v\geq u$ such that some $q\in\act{\c_v}$ satisfies $\out{q}$. Then, since the controller has at least one task that is unassigned, rule~\intrule{11} is enabled (with the controller propagating the $\false$ output) and there exists a configuration $\c'$ reachable from $\c_v$ such that $\card{Q_{\true}}{\c_v} > \card{Q_{\true}}{\c'}$. Note that the controller never propagates the $\true$ output for $\c_v$; it follows that $\card{Q_{\true}}{\c_v} \geq \card{Q_{\true}}{\c^*}$ for any $\c^*$ reachable from $\c_v$. Thus, by fairness, this implies that there exists a configuration $\c''$ reachable from $\c'$ such that $\card{Q_{\true}}{\c''} = 0$. So, some configuration $\c_\tau$ such that $O(\c_\tau) = \false$ is reached. Any configuration $\c_{\tau'}$ reachable from $\c_\tau$ also satisfies $O(\c_\tau) = \false$.
\end{proof}

\end{document}